\documentclass{llncs}        
\usepackage[T1]{fontenc}
\usepackage{amssymb}
\usepackage{amsmath}

\usepackage{amsthm}
\usepackage{proof}
\usepackage{stmaryrd}

\usepackage{floatflt} 

\usepackage{graphicx}
\usepackage{listings}
\usepackage{bussproofs}
\usepackage{tikz}

\usepackage{pgf-umlsd}
\usetikzlibrary{arrows,shapes,trees,positioning,decorations,shapes}
\usetikzlibrary{decorations.markings,automata}
\newcommand{\RN}[1]{%
  \textup{\uppercase\expandafter{\romannumeral#1}}%
}

\DeclareOldFontCommand{\rm}{\normalfont\rmfamily}{\mathrm}
\DeclareOldFontCommand{\sf}{\normalfont\sffamily}{\mathsf}
\DeclareOldFontCommand{\tt}{\normalfont\ttfamily}{\mathtt}
\DeclareOldFontCommand{\bf}{\normalfont\bfseries}{\mathbf}
\DeclareOldFontCommand{\it}{\normalfont\itshape}{\mathit}
\DeclareOldFontCommand{\sl}{\normalfont\slshape}{\@nomath\sl}
\DeclareOldFontCommand{\sc}{\normalfont\scshape}{\@nomath\sc}
\DeclareRobustCommand*\cal{\@fontswitch\relax\mathcal}
\DeclareRobustCommand*\mit{\@fontswitch\relax\mathnormal}

\colorlet{keywordcolor}{blue!50!black}
\colorlet{commentcolor}{green!60!black}
\colorlet{typecolor}{violet}
\newcommand{\sourcefont}{\ttfamily\small}
\newcommand{\commentfont}{\slshape\rmfamily\color{commentcolor}}

\lstdefinelanguage{ABS}{
        keywords={do,assert,this,new,data,type,def,case,of,local,class,interface,
        extends,implements,if,then,else,await,get,Fut,return,skip,while,module,
        import,export,from,to,suspend,delta,adds,modifies,removes,original,productline,
        features,core,corefeatures,optionalfeatures,after,when,product,hasAttribute,
        hasMethod,hasField,hasInterface,uses,root,extension,group,allof,oneof,require,
        stateupdate,objectupdate,classupdate,fi,
        exclude,original,ifin,ifout,opt,null,
        newgroup,data,thiscomp,in,joins,leaves,subtypeOf,wait,acquire,except,as,component,Pre,Abs
        },
        keywordstyle=\color{keywordcolor}\bfseries\sffamily,
        morekeywords=[2]{Unit, Int, Bool, Rat, List, Set, Pair, Fut, Maybe, String, Triple, Either, Map},
        keywordstyle=[2]\color{typecolor},
        sensitive=true,
        comment=[l]{//},
        morecomment=[s]{/*}{*/},
        morestring=[b]"
}

\lstdefinelanguage[v9]{Java}[]{Java}{
        morekeywords={module,requires,provides,uses,with,to,exports}
}
\lstdefinelanguage[ContextJ]{Java}[]{Java}{
        morekeywords={layer,with,without,proceed,before,after}
}
\lstdefinelanguage[FOP]{Java}[]{Java}{
        morekeywords={refines,original,Super}
}
\lstdefinelanguage[JastAdd]{Java}[]{Java}{
        morekeywords={aspect,syn,inh,lazy}
}

\lstdefinestyle{code}{
        basicstyle=\sourcefont\upshape,
        keywordstyle=\color{keywordcolor}\bfseries\sffamily,
        commentstyle=\commentfont,
        columns=fullflexible,
        mathescape=true,
        escapechar={\#},
        keepspaces=true,
        showstringspaces=false,
        aboveskip=8pt, 
        numbers=left,
        stepnumber=1, 
        numberstyle=\ttfamily\scriptsize\color{gray},
        numbersep=4pt,
        xleftmargin=1.5em,
        xrightmargin=1.5em,
        framexleftmargin=1.2em,
        framexrightmargin=1em,
        framextopmargin=0.5ex,
        breaklines=true,
        breakindent=3pt,
}

\lstdefinestyle{abs}{
        style=code,
        language=ABS,
}
\lstdefinestyle{java}{
        style=code,
            language=Java
}
\lstdefinestyle{java9}{
        style=code,
            language=[v9]Java
}
\lstdefinestyle{aspectj}{
        style=code,
        language=[AspectJ]Java
}
\lstdefinestyle{jastadd}{
        style=code,
        language=[JastAdd]Java
}
\lstdefinestyle{contextj}{
        style=code,
        language=[ContextJ]Java
}
\lstdefinestyle{FOP}{
        style=code,
        language=[FOP]Java
}
\lstdefinestyle{scala}{
        style=code,
        language=Scala,
        morekeywords={self}
}

\newcommand{\code}[2][]{\lstinline[style=code,basicstyle=\ttfamily\upshape,#1]{#2}}

\newcommand{\abs}[2][]{\code[style=abs,#1]{#2}}

\lstnewenvironment{srccode}[1][]{
        \minipage{\linewidth}
        \lstset{style=code,
        backgroundcolor=\color{white},
        rulecolor=\color{gray!50},
        frame=tblr,
        captionpos=b,
        #1}
}{\endminipage}

\lstnewenvironment{abscode}[1][]{
        \minipage{1\linewidth}
        \lstset{style=abs,
        backgroundcolor=\color{white},
        rulecolor=\color{gray!50},
        frame=tblr,
        captionpos=b,
        #1}
}{\endminipage}

\lstnewenvironment{javacode}[1][]{
        \minipage{\linewidth}
        \lstset{style=java,
        backgroundcolor=\color{gray!5},
        rulecolor=\color{gray!50},
        frame=tblr,
        captionpos=b,
        #1}
}{\endminipage}

\lstnewenvironment{jastaddcode}[1][]{
        \minipage{\linewidth}
        \lstset{style=jastadd,
        backgroundcolor=\color{gray!5},
        rulecolor=\color{gray!50},
        frame=tblr,
        captionpos=b,
        #1}
}{\endminipage}

\newcommand{\many}[1]{\overline{#1}}

\newcommand{\event}{\ensuremath{\mathsf{ev}}\xspace}
\newcommand{\coreactor}{\ensuremath{\mathsf{Async}}\xspace}

\newcommand{\future}{\ensuremath{f}}

\newcommand{\modal}[1]{\ensuremath{[#1]}}

\usepackage{ifthen}
\usepackage{xspace}
\usepackage{pdfpages}
\usepackage{todonotes}

\newtheoremstyle{mystyle}
  {\topsep} 
  {\topsep} 
  {} 
  {} 
  {\bfseries} 
  {\newline} 
  {.5em} 
  {} 
\theoremstyle{mystyle}

\makeatletter
\def\moverlay{\mathpalette\mov@rlay}
\def\mov@rlay#1#2{\leavevmode\vtop{%
   \baselineskip\z@skip \lineskiplimit-\maxdimen
   \ialign*{\hfil$\m@th#1##$\hfil\cr#2\crcr}}}
\newcommand{\charfusion}[3][\mathord]{
    #1{\ifx#1\mathop\vphantom{#2}\fi
        \mathpalette\mov@rlay{#2\cr#3}
      }
    \ifx#1\mathop\expandafter\displaylimits\fi}
\makeatother

\let\temp\phi
\let\phi\varphi
\let\varphi\temp

\makeatletter
\newcommand{\xRightarrow}[2][]{\ext@arrow 0359\Rightarrowfill@{#1}{#2}}
\makeatother
\newcommand{\dpp}{\ensuremath{\!:\!}\xspace}

\newcommand{\project}[3]{#1\! \upharpoonright_{#2} #3 \xspace }
\newcommand{\projectphi}[2]{#1@#2}

\newcommand{\ROLE}[1]{\ensuremath{\mathsf{#1}}}

\newcommand{\GInterOn}[4]{\ROLE{#1}\! \xrightarrow{#2} \!\ROLE{#3} \dpp #4}

\newcommand{\glGet}[2]{\ensuremath{\ROLE{#1} \!\uparrow\! #2}\xspace}

\newcommand{\Put}[1]{\ensuremath{\mathsf{Put}}\xspace}
\newcommand{\PutAs}[1]{\ensuremath{\Put~~#1}}

\newcommand{\ReadAs}[1]{\mathsf{Read} \ #1}

\newcommand{\Prcl}{\mathbf{G}}
\newcommand{\G}{\ensuremath{\mathit{G}}\xspace} 
\newcommand{\LPrcl}{\mathbf{L}}
\newcommand{\T}{\ensuremath{\mathit{L}}\xspace} 

\newcommand{\invocev}{\ensuremath{\mathsf{iEv}}\xspace}
\newcommand{\invocrev}{\ensuremath{\mathsf{iREv}}\xspace}
\newcommand{\resolvev}{\ensuremath{\mathsf{fEv}}\xspace}
\newcommand{\resolvrev}{\ensuremath{\mathsf{fREv}}\xspace}

\newcommand{\noev}{\ensuremath{\mathsf{noEv}}\xspace}
\newcommand{\ev}{\ensuremath{\mathit{ev}}\xspace}
\newcommand{\histype}{\ensuremath{\mathbb{C}}\xspace}
\newcommand{\act}{\ensuremath{\mathsf{ac}}\xspace}
\newcommand{\actC}{\ensuremath{\mathsf{A}}\xspace}
\newcommand{\prop}{\ensuremath{\mathsf{prp}}\xspace}

\newcommand{\cased}[1]{\ensuremath{
\left\{\begin{array}{ll}
#1
\end{array}\right.\\
}\xspace}
\newcommand{\casedin}[1]{\ensuremath{
\left\{\begin{array}{ll}
#1
\end{array}\right.
}\xspace}

\newcommand{\casedsclosed}[1]{\ensuremath{
\left\{\begin{array}{l}
#1
\end{array}\right\}\\
}\xspace}

\allowdisplaybreaks[1]

\newcommand{\nnn}{\ensuremath{~.~}\xspace}

\newcommand{\self}{\ensuremath{\mathbf{self}}\xspace}
\newcommand{\result}{\ensuremath{\mathbf{result}}\xspace}

\newcommand{\ADL}{\ensuremath{\mathsf{ADL}}\xspace}

\newcommand{\obs}{\ensuremath{\mathsf{O}}\xspace}

\newcommand{\COMMENT}[1]{}

\newcommand{\EK}[1]{{\color{black}{#1}}}

\newcommand{\kw}[1]{\mathsf{#1}}

\newcommand{\var}{\xabs{x}}

\newcommand{\term}{\mathsf{t}}

\newcommand{\sep}{  \ | \ }

\newcommand{\actor}{\objname}

\newcommand{\method}{\mathsf{M}}
\newcommand{\field}{\mathsf{fl}}

\newcommand{\methodfun}[1]{\mathsf{#1}}

\newcommand{\xabs}[1]{\text{\abs{#1}}}

\newcommand{\receiveTyped}[2]{ ? #1\obligation{#2}}
\newcommand{\receiveUntyped}[1]{ ? #1}
\newcommand{\sendTyped}[4]{\ROLE{#1} !_{#2} #3\obligation{#4}}
\newcommand{\sendUntyped}[3]{\ROLE{#1} !_{#2} #3}
\newcommand{\kend}{\kw{End}}
\newcommand{\skipT}{\ensuremath{\mathsf{skip}}\xspace}

\newcommand{\offer}{\&}
\newcommand{\select}{\oplus}

\newcommand{\rulename}[1]{\textbf{\scriptsize(\textsf{#1})}}

\newcommand{\INFER}[3]{\begin{array}{c}\rulename{#1} \\[1mm] \frac{\begin{array}{c}\displaystyle{#2}
\\[0.5mm]
\end{array}
}{
\begin{array}{c}
\\[-3.5mm]
\displaystyle{#3}
\end{array}
}\end{array}}

\newcommand{\TINFER}[3]{\begin{array}{c}\rulename{#1}  \frac{\begin{array}{c}\displaystyle{#2}
\\[0.5mm]
\end{array}
}{
\begin{array}{c}
\\[-3.5mm]
\displaystyle{#3}
\end{array}
}\end{array}}

\newcommand{\hastype}{:}
\newcommand{\has}{\rhd}
\newcommand{\proves}{\vdash}

\newcommand{\caus}{\mathbb{G}} 

\newcommand{\mainO}{\mathbf{main}}

\newcommand{\getABS}{\text{\abs{get}}}
\newcommand{\ifABS}{\text{\abs{if}}\ }
\newcommand{\elseABS}{\text{\abs{else}}\ }
\newcommand{\thenABS}{\text{\abs{then}} \ }
\newcommand{\fiABS}{\text{\abs{fi}}}
\newcommand{\whileABS}{\text{\abs{while}}\ }

\newcommand{\returnABS}{\text{\abs{return}}}

\newcommand{\expression}{e}
\newcommand{\statement}{\text{\abs{s}} }

\newcommand{\predicate}{\mathsf{p}}
\newcommand{\logicfunction}{\mathsf{f}}

\usepackage{scalerel,stackengine}
\stackMath
\newcommand\reallywidehat[1]{%
\savestack{\tmpbox}{\stretchto{%
  \scaleto{%
      \scalerel*[\widthof{\ensuremath{#1}}]{\kern-.6pt\bigwedge\kern-.6pt}%
          {\rule[-\textheight/2]{1ex}{\textheight}}
            }{\textheight}%
}{0.5ex}}%
\stackon[1pt]{#1}{\tmpbox}%
}

\newcommand{\postConditions}{\ensuremath{\mathsf{Post}}\xspace}

\newcommand{\osynt}[1]{\mathsf{#1}}
\newcommand{\synt}[1]{\xabs{#1}}
\newcommand{\obligation}[1]{\langle#1\rangle}

\newcommand{\methodname}{\ensuremath{\synt{m}}\xspace}

\newcommand{\type}{\ensuremath{\synt{T}}\xspace}
\newcommand{\objname}{\ensuremath{\osynt{X}}\xspace}

\newcommand{\prgm}{\ensuremath{\osynt{Prgm}}\xspace}

\newcommand{\objlang}{\ensuremath{\osynt{O}}\xspace}
\newcommand{\expr}{\ensuremath{\mathsf{e}}\xspace}

\newcommand{\transition}{\ensuremath{\rightarrow}\xspace}

\newcommand{\initial}{\ensuremath{\mathbb{I}}\xspace}

\newcommand{\lstore}{\ensuremath{\sigma}\xspace}
\newcommand{\store}{\ensuremath{\rho}\xspace}

\newcommand{\task}[1]{\ensuremath{\mathbf{prc}(#1)}}
\newcommand{\object}[1]{\ensuremath{\mathbf{ob}(#1)}}

\newcommand{\config}{\ensuremath{\osynt{C}}\xspace}
\newcommand{\trace}{\ensuremath{\mathbf{tr}}\xspace}

\newcommand{\eval}[1]{\ensuremath{ \llbracket #1 \rrbracket_{\sigma,\rho} \xspace}}
\newcommand{\retvalue}{\ensuremath{\mathsf{val}}}
\newcommand{\lastput}{\ensuremath{\mathsf{lastTerm}}}
\newcommand{\firstput}{\ensuremath{\mathsf{firstTerm}}}
\newcommand{\isput}{\ensuremath{\mathsf{res}}}

\title{Stateful Behavioral Types for ABS}
\author{Eduard Kamburjan and Tzu-Chun Chen}
\institute{Department of Computer Science, Technische Universit{\"a}t Darmstadt, Germany\\
  \email{\texttt{kamburjan@cs.tu-darmstadt.de}, \texttt{tc.chen@dsp.tu-darmstadt.de}}}

\begin{document}
\maketitle
\begin{abstract}
It is notoriously hard to correctly implement a multiparty protocol
which involves asynchronous/concurrent interactions
and constraints on states of multiple participants.
To assist developers in implementing such protocols, 
we propose a novel specification language 
to specify interactions within multiple object-oriented actors and the side-effects on heap memory of those actors. 
A behavioral-type-based analysis is presented for type checking.
Our specification language formalizes a protocol as a 
\textit{global} type, 
which describes the procedure of asynchronous method calls, 
the usage of \emph{futures}, 
and the heap side-effects with a first-order logic.
To characterize runs of instances of types, we give a model-theoretic semantics for types and translate them into logical constraints over traces.
We prove protocol adherence: 
If a program is well-typed w.r.t. a protocol, then every trace of the program adheres to the protocol, i.e., every trace is a model for the formula of the protocol's type.
\end{abstract}

\thispagestyle{plain}

\section{Introduction}\label{sec:intro}
The combination of actors~\cite{actor} with futures~\cite{Baker:1977:IGC:800228.806932} in object-oriented languages (e.g., Scala~\cite{ScalaLang} and ABS~\cite{abs}),
sometimes called \emph{Active Objects}~\cite{Active}, is an active research area for system models
and is frequently used in practice~\cite{scalamix}.
Processes of Active Objects 
communicate internally within an object via the object's heap memory.
External communication works 
via asynchronous method calls
with futures: constructs for synchronizing executions invoked by those calls.
Encapsulated heap memory and explicit synchronization points  
make it easy to locally reason about Active Objects, but hard to specify and verify \emph{global} protocols.

The main obstacle is to bridge the gap between local perspectives of single objects and global perspectives of the whole system.
As Din and Owe~\cite{Din14} pointed out, it is non-trivial to precisely specify the communication within an object's heap memory from a global perspective~\cite{noc}.
Multiparty session types (short as MPST)~\cite{Honda1}, one important member of behavioral types~\cite{Ancona:2016:BTP:3052281,GayVWY17},
are established theories for typing \textit{globally} 
\textit{stateless} concurrent interactions (i.e., method calls) 
among multiple participants (i.e., objects) to 
ensure communication safety. 
Current works in MPST~\cite{contract1,TONINHO201761}
have attempted to specify state in communication
by using global values and assuming channels as the only communication concept.
Global values are not 
sufficient to specify the non-trivial interplay of processes when taking the heap memory inside of an object into account.
Furthermore, channels are not able to fully represent the usage of futures, 
because futures, unlike channels, could expose some inner state of their object.
Namely, it exposes that the computing process has terminated and the object was inactive before and after.


We integrate the stateful analysis and specification of traces of Din et al.~\cite{Din14} into MPST,
where local verification of the endpoints compositionally guarantees the global specification of the whole system.
Functional properties are specified as a part of the communication pattern.
We ensure 
that from the perspective of each actor, its trace is not distinguishable from the global specification 
and
that the whole system is deadlock free. 

We specify passed data and modifications of heap memory 
with first-order logic (FOL) formulas 
and transform behavioral types into logical constraints on traces.
Moreover, from the model-theoretic perspective, we define \emph{protocol adherence}
as the property that every generated trace of a well-typed($\vdash$) program is a model($\models$) for the translation of the type.
The running example below illustrates the challenges for protocols in Active Objects.
\begin{floatingfigure}[l]{0.3\textwidth}
\centering\vspace{-2mm}
\includegraphics{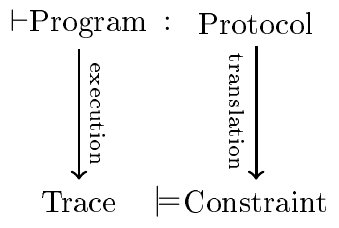}\hspace{-2mm}
\vspace{-4mm}
\end{floatingfigure}
Consider a GUI \abs{U}, a computation server \abs{S}, and an interface server \abs{I}
such that \abs{U}, without knowing \abs{S}, wants to compute some data by sending it to \abs{I} via a method call.
After executing this call, \abs{U} prepares for the next action by setting field $\xabs{intern}$ to value $\xabs{expect}$ and terminating its process to stay responsive.
\abs{I} delegates \abs{U}'s task to \abs{S} and remains responsive to other requests without waiting for \abs{S}'s computation
by invoking another method on \abs{U} with future \var{}, which will carry the computation result, back to \abs{U}.
The code and figure below implement this scenario:

\noindent\begin{minipage}{0.45\textwidth}
\begin{abscode}[basicstyle=\scriptsize]
object U{ 
  TState intern = init; 
  Int resume(Fut#$\langle$#Int#$\rangle$# x){ 
    if( this.intern!=expect } return -1;
    Int r = x.get; return r; }
  Unit start(Int j){
    Fut#$\langle$#Unit#$\rangle$# f = I!cmp(j);
    this.intern = expect; }}
\end{abscode}
\end{minipage}
\begin{minipage}{0.5\textwidth}
\begin{abscode}[basicstyle=\scriptsize, firstnumber=9]
object I{
 Unit cmp(Int dat){
   Fut#$\langle$#Int#$\rangle$# f = S!cmp(dat);
   Fut#$\langle$#Int#$\rangle$# f' = U!resume(f);}}

object S{ Int cmp(Int i){ ... }}

main { U!start(20); }
\end{abscode}
\end{minipage}
\begin{floatingfigure}[r]{0.28\textwidth}
\hspace{-10mm}
\begin{tikzpicture}[scale=0.7, every node/.style={transform shape}]

    \node [copy shadow={draw=black!30,fill=black!30,shadow xshift=0.25ex,
        shadow yshift=-0.25ex},fill=white,draw=black,thick,font=\bfseries]
        at (-1,5.5) {\large\abs{U}};

    \node [copy shadow={draw=black!30,fill=black!30,shadow xshift=0.25ex,
        shadow yshift=-0.25ex},fill=white,draw=black,thick,font=\bfseries]
        at (1,5.5) {\large\abs{I}};

    \node [copy shadow={draw=black!30,fill=black!30,shadow xshift=0.25ex,
        shadow yshift=-0.25ex},fill=white,draw=black,thick,font=\bfseries]
        at (3,5.5) {\large\abs{S}};

    \draw[dotted] (-1,2.5) -- (-1,5.2);
    \draw[dotted] (1,2.5) -- (1,5.2);
    \draw[dotted] (3,2.5) -- (3,5.2);

    \filldraw[draw=black, fill=black!50] (-0.75,4) rectangle (-1.25,5);

    \filldraw[draw=black, fill=black!50] (-0.75,3.5) rectangle (-1.25,2.5);

    \filldraw[draw=black, fill=black!50] (0.75,3.5) rectangle (1.25,4.5);

    \filldraw[draw=black, fill=black!50] (2.75,4) rectangle (3.25,3.25);

    \draw[->] (-1.75,5) -- (-1.25,5);
    {\draw[dashed,->] (2.75,3.25) -- (-0.75,2.75) node[draw=none,fill=none,font=\small,midway,below] {\large{(\abs{get})}};}

    {\draw[->] (-0.75,4.75) -- (0.75,4.5) node[draw=none,fill=none,midway,below] {\large\texttt{cmp}};}
    {\draw[->] (1.25,4.25) -- (2.75,4) node[draw=none,fill=none,midway,below] {\large\texttt{cmp}};}
    {\draw[->] (0.75,4) -- (-0.75,3.5) node[draw=none,fill=none,midway,below,inner sep=10pt] {\large\texttt{resume}};}
\end{tikzpicture}
\vspace{-2mm}
\end{floatingfigure}

In the code, \xabs{!} denotes a non-blocking call, \abs{I!cmp} calls method \abs{cmp} of \abs{I}, 
\abs{U!start} calls \abs{U.start}, \abs{U!resume} calls method \abs{resume} for continuation, 
and \abs{S!cmp} starts the actual computation at \abs{S}.
The challenge 
for formal specifications
is to express 
 that 
 (1) \abs{I} is transparent to \abs{U} and \abs{S}
 such that 
 \abs{I} must pass the same data to \abs{S} that it received from \abs{U}, and
 \abs{I} does not read the return value from \abs{S};
 and (2) \abs{U} changes its heap to \abs{expect} and reads the correct future.

\noindent
\textit{Contributions.} 
We propose (1) a specification language for actors' behaviors, that integrates FOL to specify heap memory, 
(2) model-theoretic semantics for protocol adherence, and
(3) a static type system integrating a FOL validity calculus, which guarantees protocol adherence and deadlock freedom.

\noindent
\textit{Roadmap.} Section~\ref{sec:concept} provides an overview of our approach. 
Section~\ref{sec:lang} introduces a core language for Active Objects, \coreactor, 
and its dynamic logic, Section~\ref{sec:types} gives the types and operations on them and
Section~\ref{sec:type} gives the type system. 
Section~\ref{sec:rep} extends the concept to repetition. 
Section~\ref{sec:conc} concludes and discusses related work.

\section{Scope, Challenges and an Overview of the Workflow}\label{sec:concept}

\EK{
    We aim to specify and verify \emph{session-based} systems. A session-based system is a system which has a fixed, finite set of participating objects. Each object has an assigned \emph{role} within the protocol of a session.
    Our analysis is fully static and is aimed at \emph{system validation}: Ensuring that an existing system follows a certain specification.
}

We consider object-oriented actors, which use method calls, futures, and heap memory for communication.
Every method call is asynchronous and starts a new process at the callee object.
At each such call, the active \emph{caller}
obtains a \emph{fresh} future identity, on which one may synchronize on the termination of the started process. 
An object may only switch its active process to another process if the currently active process terminates. 
The usage of futures provides programmers with the control of \textit{when} synchronize -- however, combining futures with object-oriented actors leads to the following complications:
\begin{description}
\EK{\item[Protocols with State]
In an object-oriented setting, one must take the heap memory into account when reasoning about concurrent computations. 
For one, the heap memory influences the behavior of objects. 
For another, changes of the heap memory (among coordinated actors) are not only a by-effect of communication but often the \emph{aim} of a protocol.
Actors enforce strong encapsulation and restrict communication between object to asynchronous method calls and future reads -- coordinated memory changes must be part of the specification. 
}
\item[Unexposed State] 
In the Active Object concurrency model, each process has exactly one future. Thus reading from a future is synchronizing with an unknown process \emph{and depends on the state of the process's object}. 
To avoid deadlocks, futures cannot be analyzed in isolation --- reading from a future must take the unexposed state of the object into account.
\item[Mixed Communication Paradigms] 
\EK{Processes inside an object communicate through the heap memory. This kind of communication} is hard to describe with data types, as it requires fine-grained specification of computation and has no explicit caller or callee. 
Thus, it is difficult to isolate the parts of the program which realize the communication protocol.
Furthermore, method calls are asynchronous, while future reads are synchronous.
\item[Two-Fold Endpoints] 
In the Active Object model, the callee endpoints of methods calls are \emph{objects}, but the caller endpoints and the endpoints for future synchronization are \emph{processes}. 
The interplay of multiple objects, which contain multiple processes, 
must be captured in the analysis by a two-fold notion of endpoints such that objects and processes are both endpoints. 
\end{description}
In the following, we use the example from Section~\ref{sec:intro} to show how our approach works and addresses these issues.
\paragraph{Example 1: Specifying global types.}\setcounter{example}{1}\label{ex:introex}
\vspace{-2mm}
Our specification language for side-effects is a FOL for specifying \textit{local} memory instead of global values
since (1) global values are not natural in an Active Object setting, 
and (2) a logic over memory locations (variables and fields) allows us to use a well-established theory of first order dynamic logic~\cite{Harel:1979:FDL:539855} to capture the semantics of methods. 
We formalize the scenario in Section~\ref{sec:intro} by the following global type in our specification language:
\begin{align*}
\Prcl = & \GInterOn{\mainO}{}{U}{\synt{start}\obligation{\xabs{U}.\synt{state} \doteq \synt{expect}}} \nnn
 \GInterOn{U}{}{I}{\synt{cmp}\obligation{\top,\top}} \nnn\\[-2mm]
&\GInterOn{I}{\xabs{f}}{S}{\synt{cmp}\obligation{\xabs{i}\doteq\xabs{dat},\result > 0}} \nnn
\GInterOn{I}{}{U}{\synt{resume}\obligation{\xabs{x} \doteq \xabs{f},\top}} \nnn
 \glGet{U}{\xabs{x}} \nnn \kend
\end{align*}\normalsize
We formally define the above syntax in Section~\ref{sec:types} and only give the intuition here:
$\top$ denotes true.
$\GInterOn{U}{}{I}{\synt{cmp}}$ denotes a message $\synt{cmp}$ from \ROLE{U} to \ROLE{I}, i.e., the call to a method $\synt{cmp}$.
Formula $\ROLE{U}.\synt{state} \doteq \synt{expect}$ is the postcondition for the process \emph{started} by this call at the callee object.
If two formulas are provided, the first is the precondition describing the state of the caller and the second is the postcondition describing the state of the callee and the return value, which is denoted by keyword $\result$.
The annotation $\xabs{f}$ denotes the memory location where the future of the denoted call is stored. 
Formula  $\xabs{i}\doteq\xabs{dat}$ states that $\xabs{dat}$, the parameter of \ROLE{S}\abs{.cmp}, carries the same value as received by \ROLE{I}\abs{.cmp} on parameter $\xabs{i}$, 
while formula $\xabs{x}\doteq\xabs{f}$ requires that 
parameter $\xabs{x}$ of the call at method \textsf{resume} carries the future of the previous call to \textsf{cmp}. 
Finally, $\glGet{U}{\xabs{x}}$ describes a read of \abs{U} on the future stored in the location $\xabs{x}$.
Note that we specify locations in formulas 
and avoid a situation 
 where an endpoint must guarantee an obligation containing values that it cannot access. 
Other approaches (e.g., Bocchi et al.~\cite{contract1}) allow this situation and thus require additional analyses of
history-sensitivity and temporal-satisfiability.

For the analysis, 
we adopt an approach similar to MPST:
We project a global type on endpoints defined inside it, to automatically derive local specifications for all objects and methods.
Additionally, formulas, which are used to specify conditions on the heap memory, 
are projected on the logical substructure of the callee, 
because the callee cannot access the caller's fields. 
\paragraph{Two-phase Analysis.}
The analysis requires that the protocol is encoded as a global type, 
which defines the order of method calls and future reads between objects, annotated with FO specifications of heap memory and passed data.
Our analysis has two phases. 
In Phase 1, the global type is used to generate local types for all endpoints. 
In Phase 2, the endpoints are type checked against their local types
and a causality graph is generated for checking for deadlocks.
The workflow of Phase 1 is based on MPST's approach, but is adjusted to the Active Object concurrency model:

\paragraph{Phase 1.}
The workflow of Phase 1 is shown in Fig.~\ref{fig:work}.
\begin{figure}[t]
\setlength{\abovecaptionskip}{0pt}
\centering
\includegraphics[scale=1]{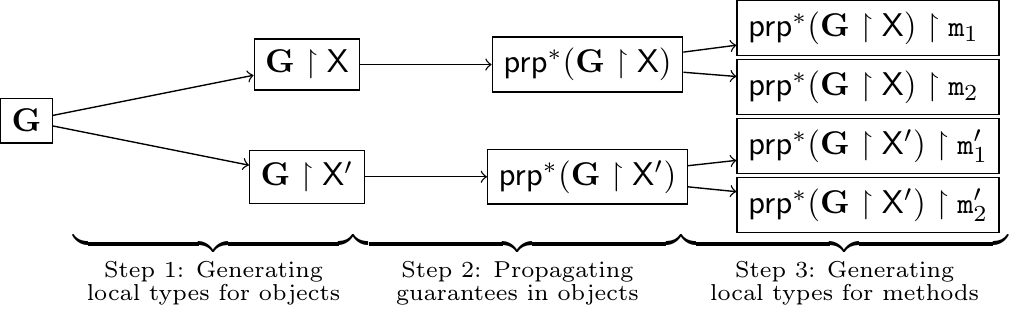}
\caption{Workflow for Phase 1: $\Prcl$ is a global type and $\upharpoonright$ denotes projection on object \objname resp. method \methodname. Function $\prop^\ast$ is the function propagating guarantees.}
\label{fig:work}
\end{figure}
\begin{itemize}
\item
\textit{Step 1}: The global type is projected onto the participating objects and generates \emph{object} types. 
Such a type specifies the obligation of an object for running methods in a certain order, 
and for guaranteeing the FOL specifications of the object's state. 
During projection, the FO-specifications are 
projected onto the substructure of the object in question. 
\item 
\textit{Step 2}: FO-specifications are propagated within an object type: as the order of method executions is specified by the specification, the postcondition of a method can be assumed as a \emph{pre}condition for the next method.
\item  
\textit{Step 3}: An object type is projected on its methods, producing \emph{method} types.
\end{itemize}

A global type encodes the following obligations (short as Obl. ) for the implementation: 
(\textit{Obl.\ a}) for each object, the observable events (calls and reads) are ordered as specified in the global type, 
(\textit{Obl.\ b}) for each method, the observable events are ordered as specified in the local type derived from the global type and
(\textit{Obl.\ c}) the whole system does not deadlock,
and adhere to the FO-specifications.

In the following, we demonstrate 
the workflow of Phase 1 for the global type in Example~\ref{ex:introex}. 
We do not formally introduce the syntax at this point.

\paragraph{Step 1: Object Types.} 
Projecting $\Prcl$ from Example~\ref{ex:introex} on object \xabs{U} results in
\small\vspace{-2mm}
\begin{align*}
\receiveTyped{\synt{start}}{\top} .
\sendTyped{I}{}{\synt{cmp}}{\top} .
\PutAs{\synt{state} \doteq \synt{expect}}. 
\receiveTyped{\synt{resume}}{\exists f.~\var \doteq f} .
\ReadAs{\var} .
\PutAs{\result > 0} 
\end{align*}\normalsize
Type  $\receiveTyped{\synt{start}}{\top}$ 
denotes a starting point for runtime execution.
Type $\sendTyped{I}{}{\synt{cmp}}{\top}$ denotes
an invocation of method $\synt{cmp}$.
Type $\PutAs{\phi}$ specifies the termination of the currently active process in a state where $\phi$ holds.
Position and postcondition of $\PutAs{\synt{state} \doteq \synt{expect}}$ are automatically derived.
The position is just \emph{before the next} method start and the postcondition is taken from the call in the global type.
The analysis ensures that no method executes in-between.
The precondition of \synt{resume} is \emph{weakened}, since
field $\xabs{f}$ is not visible to \abs{U} and callee \abs{U} cannot use all information from caller \abs{I}.
Weakening ensures that all locations in $\phi$ are visible to \abs{U}. 
Type $\ReadAs{\var}$ specifies a synchronization on the future stored in $\var$. 

\paragraph{Step 2: Propagation.} 
In the next step we propagate the postcondition of the last process to the precondition of the next process.
No process is specified as active between \abs{start} and \abs{resume}, so the heap is not modified --- thus, the postcondition of \abs{start} still holds when \abs{resume} starts.
Adding $\synt{state} \doteq \synt{expect}$ to the precondition of \abs{resume} strengthens the assumption for the type checking of \abs{resume}.
The propagation of conditions results in:\small\vspace{-2mm}
\begin{align*}
\prop^{\ast}(\project{\Prcl}{}{\xabs{U}}) =& \receiveTyped{\synt{start}}{\top} \nnn
\sendTyped{I}{}{\synt{cmp}}{\top} \nnn
\PutAs{\synt{state} = \synt{expect}}\nnn \\[-1mm]
&\receiveTyped{\synt{resume}}{\exists f.~\var \doteq f \wedge \synt{state} \doteq \synt{expect}} \nnn
\ReadAs{\var} \nnn
\PutAs{\result > 0} 
\end{align*}\normalsize

\paragraph{Step 3: Method Types.}
We generate a \emph{method type} to specify a method in isolation.
Projecting the object type in Step 2 on method \synt{resume} generates:\small\vspace{-2mm}
\begin{align*}
\prop^{\ast}(\project{\Prcl}{}{\xabs{U}})\!\upharpoonright\!\synt{resume}&=\receiveTyped{\synt{resume}}{\exists f.~\var \doteq f \wedge \synt{state} \doteq \synt{expect}} .
\ReadAs{\var} .
\PutAs{\result > 0} 
\end{align*}\normalsize

\noindent Method types share the syntax with object types. Projection from object types splits the object type
at positions where one method ends and another one starts.

\begin{figure}[b]\small
\centering
\includegraphics[scale=0.75]{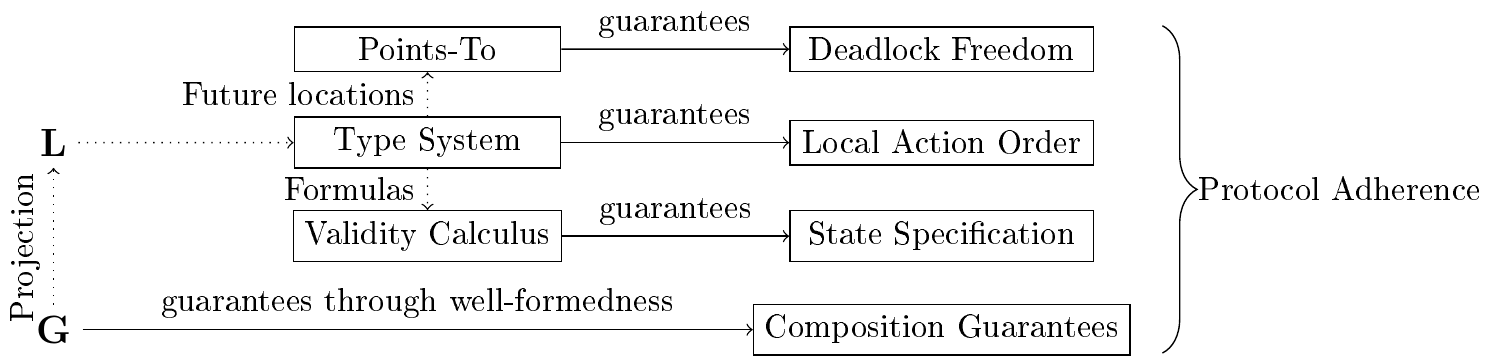}
\caption{Workflow for Phase 2 in our analysis.}
\label{fig:phase2}
\end{figure}
\paragraph{Phase 2.} After generating method types, 
Phase 2 of the 
analysis checks the implementation of methods against their method types, and checks the formulas for validity.
The type checking of method types guarantees the correct local order of events (\textit{Obl.\ b}). 
State specifications are checked by integrating a validity calculus~\cite{Din14} into the type system.
To guarantee (\textit{Obl.\ a} and \textit{c}), 
we require the following analyses:
\paragraph{Causality Graph.} \label{graph}
We generate a causality graph to ensure deadlock freedom (\textit{Obl.~c}): 
A deadlock free causality graph for Active Objects is cycle-free~\cite{MHPDead,HenrioLM17}.
A causality graph is also used to ensure that methods of one object are executed in the order specified in the global type
that the object obeys to (\textit{Obl.~a}).
\begin{center}
\begin{tikzpicture}[scale=0.8, every node/.style={scale=0.8}]
    \node[draw] at (0, 1.5)   (u) {\abs{U}};
    \node[draw,circle] at (1, 1.5)   (u1) {?};
    \node[draw,circle] at (2, 1.5)   (u2) {!};
    \node[draw,circle] at (4, 1.5)   (u3) {$\downarrow$};
    \node[draw,circle] at (7, 1.5)   (u4) {?};
    \node[draw,circle,fill=gray] at (8, 1.5)   (u5) {$\uparrow$};
    \node[draw,circle] at (9, 1.5)   (u6) {$\downarrow$};

    \node[draw] at (0, 0.75)   (i) {\abs{I}};
    \node[draw,circle] at (2.75, 0.75)   (i1) {?};
    \node[draw,circle] at (3.75, 0.75)   (i2) {!};
    \node[draw,circle] at (5.75, 0.75)   (i3) {!};
    \node[draw,circle] at (9.75, 0.75)   (i4) {$\downarrow$};

    \node[draw] at (0, 0)   (r) {\abs{S}};
    \node[draw,circle] at (4.5, 0)   (s1) {?};
    \node[draw,circle, fill=gray] at (8.5, 0)   (s2) {$\downarrow$};

     \draw[->] (u1) -- (u2);
     \draw[->] (u2) -- (u3);
     \draw[->,dotted] (u3) -- (u4);
     \draw[->] (u4) -- (u5);
     \draw[->] (u5) -- (u6);

     \draw[->] (u5) -- (s2);

     \draw[->] (i1) -- (i2);
     \draw[->] (i2) -- (i3);
     \draw[->] (i3) -- (i4);

     \draw[->] (s1) -- (s2);

     \draw[->] (u2) -- (i1);
     \draw[->] (i2) -- (s1);
     \draw[->] (i3) -- (u4);
\end{tikzpicture}
\end{center}\vspace{-2mm}
The nodes are the local types from the projected object types.
A solid edge connecting two nodes models that the statement for the first type directly causes the statement for the second type; 
for example, there are edges from a call to the corresponding receiving type.
The graph is partially generated from $\Prcl$, and partially generated from the code: 
The edge connecting the gray nodes is added by a \emph{Points-To} analysis, 
which maps a location of a future to the methods resolving this future.
\EK{
    The termination of a method causes the start of the next (as the object cannot switch the active process otherwise), but does not select the next method itself.
    A dotted edge models such \emph{indirect} causality: 
Indirect causality edges are considered when checking cycle-freedom check for deadlock freedom,
but not for checking the method order.
}

\paragraph{Model-theoretic Semantics.}
One of our contributions is the definition and verification of \emph{protocol adherence} from a model-theoretic point of view: 
The property that a program follows a specified scenario (the protocol) if every generated trace is a model for the translation of the global type.
We thus define protocol adherence through a \emph{logical} characterization of global types and translate types into constraints over \emph{traces}, which are sequences of configurations generated by the program.

\EK{This \emph{declarative} approach for defining protocol adherence allows us to connect 
the FO properties embedded in the type to the execution of methods by using a dynamic logic: 
For a statement $\statement$ the dynamic logic formula $\phi \Rightarrow [\statement]\psi$ expresses that the first-order formula $\psi$
holds after executing $\statement$, if $\phi$ holds in the beginning.
From the perspective of the trace logic, FOL describes a single configuration in the trace, while the modality $[\statement]$ relates the configuration before executing $\statement$ with the configuration
after executing $\statement$. We use modalities during type checking.}

\section{\coreactor, a Core Actor-Based Language using Futures}\label{sec:lang}
We introduce \coreactor, a simple Active Object language based on ABS~\cite{abs}. 
Due to space limitations, 
we only present the basic constructs of \coreactor below. 
For branching constructs we refer to Section~\ref{sec:branch}; repetition is introduced in Sec.~\ref{sec:rep}.
An \coreactor-program consists of a \abs{main} statement and a set of actors, which are objects that 
have fields and method but do not share state.
%
Inside an object, processes do not interleave and the currently active process must terminate before another one is scheduled. 
Therefore, single methods can be considered sequential for analysis.
We assume standard operations, literals and types for booleans, integers, lists and $\mathsf{Object}$.
\begin{definition}[\coreactor-Syntax]
Let \expr denote expressions, \type denote data types, $\var$ denote variable and field names, 
\objname denote object names, 
and \abs{Fut<T>} denote future types.
$\many{\cdot}$ represents possibly empty lists and $[\cdot]$ represents optional elements.\small\vspace{-2mm}
\begin{align*}
\prgm &::= \many{\objlang}~\xabs{main}\{\objname!\methodname(\many{\expr})\} \quad\! 
\objlang ::= \xabs{object}~\objname~\{\many{\method}~\many{\type~\var=\expr}\} \quad 
\method ::= \type~\methodname(\many{\type~\var})\{\statement;\xabs{return}~\expr\}\\
\statement &::= \big[[\xabs{Fut<}\type\xabs{>}]~\var = \big] \objname!\methodname(\many{\expr}) \sep  [\type]~\var=\expr \sep \big[\type\big]~\var = \expr\xabs.\getABS
\sep \xabs{skip} \sep\statement;\statement 
\end{align*}
\end{definition}
\vspace{-2mm}

Objects communicate only by asynchronous method calls using futures. Upon a method call, a fresh future is generated on callee side and is passed to the caller.
The callee writes the return value into the future upon termination of the corresponding process;  anyone with the access to the future can read, but not write, into it. 
\EK{We only consider static sessions, in which all objects are created before the start of the system.}
\coreactor is a standard imperative language with two additional statements:
(1) $\var = \objname ! \methodname (\many{\expr})$ calls method $\methodname$ with parameters $\many{\expr}$ on object $\objname$. The generated future is stored in $\var$.
The caller continues execution, while the callee is computing the call on $\methodname$ or scheduling 
$\methodname$ for later execution if another process is currently active.
(2) $\expr. \getABS$ reads a value from the future stored in $\expr$. If the process computing this future has not terminated, the reading process blocks.

To define a small-step reduction relation over events for the semantics of  \coreactor,
we first define an event as a process action with visible communication:
\begin{definition}[Events]
Let $\future,\future'$ range over futures. 
An \emph{event}, denoted by $\event$, is defined by the following grammar: \small \vspace{-2mm}
\begin{align*}
\event ::=  &\invocev(\objname,\objname',\future,\methodname,\many{\expr}) \mid \invocrev(\objname,\future,\methodname) \mid \resolvev(\objname,\future,\methodname,\expr) \mid \resolvrev(\objname,\future,\expr) \mid \noev
\end{align*}\normalsize
\end{definition}
\noindent
An \emph{invocation} $\invocev(\objname,\objname',\future,\methodname,\many{\expr})$ models that $\objname$ calls $\objname'.\methodname$ using $\future$ and passes $\many{\expr}$ as parameters. 
An \emph{invocation reaction} $\invocrev(\objname,\future,\methodname)$ models that $\objname$ starts executing $\methodname$ to resolve $\future$. 
A \emph{resolving} $\resolvev(\objname,\future,\methodname,\expr)$ models that $\objname$ resolves $\future$, 
which contains $\expr$ at the moment,
by finishing the execution of $\methodname$.
A \emph{fetch} $\resolvrev(\objname,\future,\expr)$ models that $\objname$ reads value $\expr$ from $\future$. 
\EK{Finally, $\noev$ models no visible communication.}

A configuration is composed of \emph{processes} and \emph{objects}.
A process has a unique future \future, a store {\lstore} which maps variables to literals, and the name \objname of its object.
An object has a unique name \objname, an active future {\future}, and a store {\store} which maps fields to literals. 
\begin{definition}[Runtime Syntax of Processes and Objects] 
The following grammar defines runtime processes and objects as configurations \config:
\small\vspace{-2mm}
\begin{align*}
\config ::= \task{\objname,\future,\methodname(\statement),\lstore} 
\sep 
\task{\objname,\future,\retvalue(\expr),\lstore} 
\sep 
\object{\objname,\future,\store} \sep \config~\config
\end{align*}\normalsize
\end{definition}
\noindent A process either is executing a method $\methodname$ for a request carried by $\future$ at some object $\objname$, represented by $\task{\objname,\future,\methodname(\statement),\lstore}$,  
or has returned $\expr$, represented by $\task{\objname,\future,\retvalue(\expr),\lstore}$.
An object $\object{\objname,\future,\store}$ has its name $\objname$, the future of the active process $\future$ and the heap $\store$.
We write $\object{\objname,\bot,\store}$ to indicate that $\objname$ is inactive.
Composition of configurations is commutative and associative, i.e., $\config~\config' = \config'~\config$ and $\config~(\config'~\config'') = (\config~\config')~\config''$.
We denote the initial configuration of a program \prgm with $\initial(\prgm)$. 
If all processes of a configuration \config have terminated, the configuration also terminates. 
The body of method \methodname is denoted by $M(\methodname)$. We write $\widehat{M}(\methodname,\many{\expr})$ for the initial local store 
of a task executing $\methodname$ with parameters $\many{\expr}$. 

We use \emph{traces}, sequences of pairs of events and configurations, to capture the behavior of a program. 
We only consider terminating runs and define big-step semantics  $\prgm \Downarrow \trace$ for \emph{finite} traces:
\begin{definition}[Run and Big-Step Semantics]
A run from $\config_1$ to $\config_n$ is a sequence of configurations $\config_1,\dots,\config_n$ with events $\event_1,\dots,\event_{n-1}$ such that:\vspace{-2mm}
\[\config_1 \transition_{\event_1} \config_2 \transition_{\event_2} \dots \transition_{\event_{n-1}} \config_n\]\normalsize
The trace $\trace$ of a run
is a sequence $(\event_1,\config_1),\dots,(\event_m,\config_m)$ where for every $1 \leq j < m \leq n$ there is a $\config$ 
such that $\config_j\!\transition_{\event_{j}}\!\!\config$ is in the run  
and $\event_j \neq \noev$. 
An \coreactor program \prgm generates \trace, 
written $\prgm \Downarrow \trace$, if 
there is a run from its initial configuration to a terminated configuration such that \trace is the trace of this run.
\end{definition}


\begin{figure}\footnotesize
\centering
\scalebox{0.9}{%
$\TINFER{call}
{
    \config \text{ does not contain}\ \future' \quad \eval{\expr} = \objname' \quad \config = \object{\objname,\future,\store}~\config'\quad \event = \invocev(\objname,\objname',\future,\methodname,\eval{\many{\expr'}})
}{
        \task{\objname,\future,\methodname(\expr!\methodname'(\many{\expr'});\statement),\lstore}~\config
        \transition_{\event}
        \task{\objname,\future,\methodname(\statement),\lstore}~\task{\objname',\future',\methodname'(M(\methodname')),\widehat{M}(\methodname,\eval{\many{\expr'}})}~\config
}
$}\vspace{2mm}
\scalebox{0.95}{%
$\TINFER{start}
{
    \event = \invocrev(\objname,\future,\methodname)
}{
        \task{\objname,\future,\methodname(\statement),\lstore}~\object{\objname,\bot,\store}~\config
        \transition_{\event}
        \task{\objname,\future,\methodname(\statement),\lstore}~\object{\objname,\future,\store}~\config
}
$}\vspace{2mm}
\scalebox{0.95}{%
$\TINFER{sync}
{
    \config = \task{\objname',\future',\retvalue(\expr'),\lstore'}~\config'\qquad \eval{\expr} = \future'\quad \event = \resolvrev(\objname,\future',\expr')
}{
                \task{\objname,\future,\methodname(\var = \expr.\getABS;\statement),\lstore} ~\object{\objname,\future,\store}~\config
                \transition_{\event}  
                 \task{\objname,\future,\methodname(\var = \expr';\statement),\lstore} ~\object{\objname,\future,\store}~\config
}
$}\vspace{2mm}
\scalebox{0.95}{%
$\TINFER{end}
{
    \event = \resolvev(\objname,\future,\methodname,\expr)
}{
                 \task{\objname,\future,\methodname(\returnABS~\expr),\lstore} ~\object{\objname,\future,\store}~\config
                 \transition_{\event}  
                 \task{\objname,\future,\retvalue(\eval{\expr}),\lstore} ~\object{\objname,\bot,\store}~\config
}
$}
\caption{The selected semantics rules. Full rules are provided in~\cite{newtechreport}.}
\label{fig:sos}
\end{figure}
Fig.~\ref{fig:sos} defines the reduction relation $\transition_{\event}$ for the semantics.  
\eval{\expr} denotes the evaluation of an expression \expr under stores $\sigma$ and $\rho$.
Rule \rulename{call} executes a method call on the object stores in $\expr$ by generating a fresh future $f'$ and an invocation event.
The new process is not set as active upon creation by \rulename{call}.
By rule \rulename{start}, the object $\objname$ must be inactive, when the process is started. An invocation reaction event is generated.
Rule \rulename{sync} synchronizes on a future $f'$ stored in $\expr$, by checking whether the configuration contains $\task{\objname',\future',\retvalue(\expr'),\lstore'}$, i.e. $f'$ is resolved, and reads the return value $\expr'$.
Rule \rulename{end} terminates a process.
In all other rules, the \event parameter is \noev.

\paragraph{Dynamic Logic.} \label{sec:ADL}
A dynamic logic combines 
FO-formulas over the heap with symbolic executions~\cite{loop,SEKing} of statements. 
A symbolic execution uses symbolic values to describe a possible set of actual values.
It does not reason about one execution of the statement, but describes a \emph{set} of executions. 
\begin{example}\label{ex:logic}
Formula 
$\exists \text{\abs{Int}}~a.~\big(a > 0 \wedge \text{\abs{i}} > a\big) \rightarrow [\text{\abs{j = i*2;}}] \text{\abs{j}} > 0$
describes that if there is a number $a$ bigger than $0$ and smaller than the value stored in \abs{i}, then after executing $\text{\abs{j = i*2;}}$, variable \abs{j} contains a positive value.
\end{example}
\noindent Based on ABSDL~\cite{absdl}, 
we present \coreactor Dynamic Logic (short as \ADL), which extends first-order logic over program variables and heap memory with modalities that model the effect of statements.
In this logic, method parameters are special variables and
a modality is a formula $[ \statement ]\phi$ which holds in a configuration, say \config, 
if $\phi$ holds in every configuration reached from \config after executing $\statement$.
We focus on the semantics of \emph{modality-free} formulas, which have configurations as models; the semantics of modalities is a transition relation between configurations.

\begin{definition}[Formulas $\phi$] \label{def:constraints}
We define the set of formulas $\phi$ and terms $\term$ 
by the following grammar, where $\predicate$ ranges over predicate symbols, 
$\logicfunction$ ranges over function symbols, 
$\mathsf{x}$ ranges over logical variables, 
and $\mathsf{v}$ ranges over logical and program variables.
The set of formulas is denoted by \ADL.\vspace{-2mm}
\begin{align*}
 \phi ::=   
\mathbf{tt} \mid \neg \phi \mid \phi \vee \phi \mid \predicate(\term \dots \term) \mid \term \geq \term \mid \term \doteq \term \mid \exists \type~ \mathsf{x}; \phi \mid \modal{\statement}\phi
\quad\quad
\term ::=     \mathsf{v} \mid \logicfunction ( \term  \dots  \term ) 
\end{align*}
\end{definition}
\vspace{-2mm}
Local program variables (i.e., $\xabs{v}$) are modeled as special function symbols. 
To model heap accesses, 
following Schmitt et al.~\cite{SchmittEtAl2011}, 
we use two function symbols $\textsf{store}$ and $\mathsf{select}$ with (at least) the axiom
\(\mathsf{select}(\textsf{store}(\textit{heap},\xabs{f},\xabs{o},\textit{value}),\xabs{f},\xabs{o}) = \mathit{value}\)
where \textit{heap} is a special local program variable modeling the heap explicitly. 
A special function symbol $\result$ is interpreted as the return value of a method, 
and a logical variable is \emph{free} if it is not bound by any quantifier. 
\begin{definition}
A formula $\phi$ is \emph{valid} if it evaluates to true in every configuration.
\end{definition}
\noindent Formulas are \emph{global} or \emph{\objname-formulas}.
Global formulas refer to
the heap of multiple objects, while \objname-formulas refer only to \objname.
The latter
contains only the function symbols for elements from \objname and the special function symbol \self modeling the reference to \objname.
For proving that an \objname-formula holds for a given state, if suffices to locally check the code of \objname.
A validity calculus for \ADL is presented in~\cite{Din14}.
\begin{definition}
Let $\phi$ be a formula. The weakened $\objname$-formula $\projectphi{\phi}{\objname}$ 
is obtained by replacing all function symbols in $\phi$ which are not exclusive to $\objname$ (i.e., refer to the fields of other objects) 
by free variables and existentially quantifying over them.
\end{definition}
\begin{example}\label{ex:logic2}
Let $\field$ be a field, \objname an object and \abs{i} the parameter of some method in class $\objname$. 
Consider $\phi = \objname.\field > 0 \wedge \text{\abs{i}} > \objname.\field$.
The formula $\phi$ is an $\objname$-formula, as $\phi = \projectphi{\phi}{\objname}$.
The weakening for some object $\objname'$ is $\projectphi{\phi}{\objname'} = \exists \text{\abs{Int}}~a.a > 0 \wedge \text{\abs{i}} > a$. 
$\projectphi{\phi}{\objname'}$ does not reason about $\objname.\field$, but still has the information of $\xabs{i}>1$.
\end{example}

\section{Behavioral-Type-Based Stateful Specification}\label{sec:types}
We define a specification language for global types to 
specify the behavior of the system. 
Following Sec.~\ref{sec:lang}, we only represent the key constructs and leave branching to Section~\ref{sec:branch} and repetition to Sec.~\ref{sec:rep}.
\begin{definition}[Syntax of Global Types] \label{def:globaltypes} \rm
Let $\phi, \psi$ range over \emph{modality-free} \ADL formulas and $\objname_i$ range over object names. 
$[\cdot]$ denotes optional elements.\vspace{-3mm}
\[
\Prcl ::= \GInterOn{\mainO}{}{\objname}{\methodname\obligation{\phi} }. \G
\qquad
\qquad
\G ::=
\GInterOn{\objname_1}{[\var]}{\objname_2}{\methodname\obligation{\phi, \psi} }. \G
\sep \glGet{\objname}{\expr}.\G
\sep \kend
\]
\end{definition}
\vspace{-3mm}
The \emph{calling type} $\GInterOn{\objname_1}{[\var]}{\objname_2}{\methodname\obligation{\phi, \psi} }$ specifies 
a method call from $\objname_1$ to $\methodname$ at $\objname_2$. 
If $\var$ is not omitted above the arrow,
the future of this call 
must be stored in location $\var$. 
The \ADL-formula $\phi$ 
specifies (1) the call parameters passed to the callee 
and (2) the memory 
of $\objname_1$ at the moment of the call.
Formula $\psi$ is the postcondition of the callee process and specifies the state of $\objname_2$ and the return value once \methodname terminates.
The exact point of termination is derived during projection.
The initial method call $\GInterOn{\mainO}{}{\objname}{\methodname\obligation{\psi} }$ 
only specifies the postcondition of the process running $\objname.\methodname$.
Type $\glGet{\objname}{\expr}$ specifies a synchronization on the future, to which the expression \expr evaluates. 
Every synchronization must be specified.
$\kend$ specifies the end of communication.

$\Prcl$ denotes a complete protocol with an initializing method call, while $\G$ denotes partial types.
Even without fields in the formula, the implementation is referenced in the specification, as endpoints are object names.
Object and method types share the same syntax. Together we call them \emph{local types}. 
The grammar of local types is defined as follows: 
\begin{definition}[Syntax of Local Types] \label{def:localtypes} \rm
Let $\phi$ range over \emph{modality-free} \ADL formulas and $\objname_i$ range over object names. 
$[\cdot]$ denotes optional elements.
\[
\LPrcl ::=  \receiveTyped{\methodname}{\phi}.\T
\qquad
\T ::=  \receiveTyped{\methodname}{\phi}.\T
\sep \sendTyped{\objname}{[\var]}{\methodname}{\phi}.\T
\sep \PutAs{\phi}.\T
\sep \ReadAs{\expr}.\T
\sep \skipT.\T
\sep  \kend
\]
\end{definition}
The type $\receiveTyped{\methodname}{\phi}$ denotes the start of a process computing $\methodname$ in a state where formula $\phi$ holds. 
Formula $\phi$ is the precondition of $\methodname$ and describes the local state and method parameters of $\methodname$.
Type $\PutAs{\phi}$ denotes the termination in a state where $\phi$ holds. Formula $\phi$ is a postcondition and describes the return value and the local store. 
Contrary to global types, a postcondition of a process is not annotated at the call, but at the point of termination \EK{because the point of termination is now explicit}.
Type $\sendTyped{\objname}{[\var]}{\methodname}{\phi}$ 
corresponds to the caller side of $\GInterOn{\objname_1}{[\var]}{\objname_2}{\methodname\obligation{\phi, \psi}}$.
Type $\ReadAs{\expr}$ 
models a read from $\expr$ and \skipT denotes no communication.
As for global types, $\kend$ models the end of communication. In our examples, we omit $\kend$ for brevity's sake.
We use $\LPrcl$ for complete local types and $\T$ for partial local types.

Projection has three steps: (1) projection of global types on objects, (2) \COMMENT{specification}condition  propagation, and (3) projection of object types on methods.

\paragraph{Projection on Objects.}
The projection on objects ensures that every object can access all locations occurring in its specification and adds $\PutAs{\phi}$ at the correct position. 
This requires an additional parameter in the projection to keep track of which process is specified to be active and what its postcondition is.

To track the postcondition of the last active method of an object, 
we use a partial function $\act:\obs \rightharpoonup \ADL$ to map objects to formulas. 
If no method was active on \objname yet, \act is undefined, written $\act(\objname) = \bot$. 
The projection of $\G$ on an object $\objname$ is denoted by $\project{\G}{\act}{\objname}$.
The selected projection rules for methods calls and termination are given in Fig.~\ref{fig:projobj}.
We write $\act_{\bot}$ for the function defined by $\forall \objname.~\act(\objname) = \bot$. For updates, we write $\act[\objname \mapsto \psi](\objname') = \casedin{ \psi &\text{ if }\objname = \objname' \\ \act(\objname') &\text{ otherwise}}$. 
\begin{figure}[tbh]\small\vspace{-2mm}
\begin{align*}
%
\rulename{1}\project{\GInterOn{\objname_1}{\var}{\objname_2}{\methodname\obligation{\phi,\psi}}. \G}{\act}{\objname_1} 
&=
  \sendTyped{\objname_2}{\var}{\methodname}{\phi} . (\project{\G}{\act[\objname_2 \mapsto \psi]}{\objname_1}) \text{ if } \act(\objname_1) \neq \bot \wedge \phi= \projectphi{\phi}{\objname_1}\\
\rulename{2}\project{\GInterOn{\objname_1}{\var}{\objname_2}{\methodname\obligation{\phi,\psi}}. \G}{\act}{\objname_2} 
&=
  \casedin{
  \receiveTyped{\methodname}{\projectphi{\phi}{\objname_2}} . (\project{\G}{\act[\objname_2 \mapsto \psi]}{\objname_1})                           & \text{if } \act(\objname_2) = \bot \\
  \PutAs{\act(\objname_2)}.\receiveTyped{\methodname}{\projectphi{\phi}{\objname_2}}. (\project{\G}{\act[\objname_2 \mapsto \psi]}{\objname_1})   & \text{if } \act(\objname_2) \neq \bot\\
  }\\  
\rulename{3}\project{\GInterOn{\objname_1}{\var}{\objname_2}{\methodname\obligation{\phi,\psi}}. \G}{\act}{\objname} 
&=
  \skipT. (\project{\G}{\act[\objname_2 \mapsto \psi]}{\objname}) \text{ if } \objname_2 \neq \objname \neq \objname_1 
  \\
\rulename{4}\project{\GInterOn{\mathbf{main}}{}{\objname_2}{\methodname\obligation{\phi}}.\G}{\act_{\bot}}{\!\!\objname_1} &=
\casedin{
\receiveTyped{\methodname}{\projectphi{\phi}{\objname_2}} . (\project{\G}{\act[\objname_2 \mapsto \psi]}{\objname_1}) 
&
\text{ if $\objname_2 = \objname_1$}
\\
\skipT. (\project{\G}{\act[\objname_2 \mapsto \psi]}{\objname_1}) 
& 
\text{ if $\objname_2 \neq \objname_1$}
}
\\ 
\rulename{5}  \project{\kend}{\act}{\objname} 
  &=
\casedin{
  \PutAs{\act(\objname)}.\kend
  & \text{ if } \act(\objname) \neq \bot\\
  \kend
  & \text{ if } \act(\objname) = \bot\\
}
\end{align*}\vspace{-6mm}
\normalsize
\caption{The selected rules for projection on objects.}
\label{fig:projobj}
\end{figure}

When projecting on caller $\objname_1$, a sending local type is generated by \rulename{1} if $\objname_1$ has an active process ($\act(\objname_1) \neq \bot$)
and the precondition can be proven by the caller ($\phi = \projectphi{\phi}{\objname_1}$). 
If the callee has an active process (\EK{i.e., the last active postcondition exists:} $\act(\objname_2) \neq \bot$), then the termination type for the active process is added by \rulename{2} before the receiving type.
If the callee is specified as being inactive (\EK{i.e., no process was running before and no postcondition is tracked} $\act(\objname_2) = \bot$), then only the receiving type is added by \rulename{2}. 
When projecting on any other object, \skipT is added by \rulename{3}.
In any case, \act is updated and maps the callee to a new postcondition.
Rules \rulename{4} and \rulename{5} are straightforward.
\EK{As usual, projection is undefined if no rule matches, and we omit $\act_{\bot}$ and write just $\project{\Prcl}{}{\objname}$.}

\paragraph{Propagation.}
In our concurrency model the heap does not change if no process is active. All guarantees from the last active process still hold for the next process.
By propagation, formulas are added from the postcondition of one method to the precondition of the next. 
Propagation moves formulas from where they \emph{must hold} to all points where they still are \emph{assumed to hold}.
Propagation 
replaces a partial local type, if the partial type matches the given pattern. 
\begin{definition}[Propagation]\label{def:prop}
The propagation function $\prop$ is defined
via term \emph{rewriting} (denoted $\rightsquigarrow$) as follows.  
$\prop^{\ast}$ denotes the fixpoint of rewriting.
\vspace{-2mm}
\[
\rulename{1}~\PutAs{\phi}.\receiveTyped{\methodname}{\psi} \rightsquigarrow \PutAs{\phi}.\receiveTyped{\methodname}{\psi \wedge \phi@\objname} \text{   where $\objname$ is the target object}
\]
\normalsize
\end{definition}

\paragraph{Projection on Methods.}
The projection on a method, denoted by 
$\T \upharpoonright_{\methodname'} \methodname$, 
results in a \emph{set} of method types. 
A method may have multiple method types, as long as the method types are \emph{distinguishable}, which means
that they have non-overlapping preconditions.
Formally, two preconditions $\phi_1$ and $\phi_2$, are distinguishable if the 
formula $\neg (\phi_1 \wedge \phi_2)$ is valid. 
\EK{In the case of overlapping preconditions, multiple preconditions can hold at the same time and it is not guaranteed that the correct type will be realized.}

The rules for projection on a method are straightforward and we refer to the Section~\ref{sec:concept} for an example 
and to the appendix for full definitions.
%
\begin{definition}[Well-Formedness]\label{def:well}
A global type $\Prcl$ is \emph{well-formed}, if the projections on all methods are defined and all types of a method are distinguishable. 
\end{definition}

\paragraph{Semantics of Types as Constraints on Traces.}
To formalize the behavioral types of the previous section, we transform them into first-order constraints over traces.

We define $\histype$ as a function transforming global types to constraints on traces.
Recall that we have defined $\config$ for configurations and $\event$ for events.
The primitive $\config(i)$ references the $i$th configuration and $\event(i)$ references the $i$th event in a trace.
We use events and formulas as colors and thus include futures, method names, literals and object names in the domain. 
Constraints refer to \ADL formulas $\phi$ with $\config(i) \models \phi$, meaning that in the $i$th configuration, $\phi$ holds.

To restrict a constraint to a subtrace, we use \emph{relativization}~\cite{henkin}, a \emph{syntactic} restriction of constraint $\gamma$ to a substructure described by another constraint $\chi$.
\begin{definition} \label{def:relative}
Let $\chi(x)$ be a constraint with a free variable $x$ of data type $\type$ and $\gamma$ another constraint.
The relativization of $\gamma$ with $\chi(x)$, written $\gamma[ x\in \type / \chi ]$, replaces all subconstraints of the form $\exists y \in \type.\gamma'$ in $\gamma$ by
$\exists y\in \type. \chi(y) \wedge \gamma'$.
%
%
\end{definition}
\noindent The main rules for translating $\Prcl$   into a constraint $\histype(\Prcl)$ are defined as follows. 
\begin{definition}[Semantics of Global Types]\label{def:msosem}
Predicate $\isput(i)$ holds if $\event(i)$ is a resolving event and $\actC(i,\objname)$ holds if \objname is active in $\config(i)$.\fontsize{8}{10} 
\begin{align*}
\rulename{1}\histype(&\GInterOn{\mainO}{}{\ROLE{\objname_2}}{\methodname \obligation{\psi}}.\G) ~=~\exists j,k .~\exists f.~\exists \expr'.~\event(j) \!\doteq\! \invocrev(\objname_2,f,\methodname) \wedge \config(j) \!\models\! \projectphi{\phi}{\objname_2} \wedge\\
&\event(k) \!\doteq\! \resolvev(\objname_2,f,\expr') \wedge \config(k) \!\models\! \psi  \wedge \forall l . l\!\neq\! j \wedge l\!\neq\! k \Rightarrow \isput(l) \wedge \histype(\G)\\
\rulename{2}\histype(&\GInterOn{\ROLE{\objname_1}}{\var}{\ROLE{\objname_2}}{\methodname \obligation{\phi,\psi}} ) ~=~\exists i,j,k .~\exists f.~\exists \expr, \expr'.~\\
&\event(i) \!\doteq\! \invocev(\objname_1,\objname_2,f,\methodname,\expr) \wedge \config(i) \!\models\! \phi \wedge \event(j) \!\doteq\! \invocrev(\objname_2,f,\methodname) \wedge \config(j) \!\models\! \projectphi{\phi}{\objname_2} \wedge\\
&\event(k) \!\doteq\! \resolvev(\objname_2,f,\expr') \wedge \config(k) \!\models\! \psi \wedge \config(i) \!\models\! (\objname_1.\xabs{x} \!\doteq\! f) \wedge \forall l .~l\!\neq\! i \wedge l\!\neq\! j \wedge l\!\neq\! k \Rightarrow \isput(l) \\
\rulename{3}\histype(&\G_1 . \G_2) =\bigwedge_{\objname}\big(\exists i\in\mathbb{N}.~\histype(\G_1)[j\in\mathbb{N}/\actC(j,\objname) \Rightarrow j < i] \wedge \histype(\G_2)[j\in\mathbb{N}/\actC(j,\objname) \Rightarrow j \geq i]\big)
\end{align*}\normalsize
\end{definition}
\vspace{-3mm}
The constraint \rulename{1} for the call type has three events modeling (1) a call, (2) the start of the process and (3) the existence of the termination of the process. 
Moreover, the projected formulas hold at the configurations for these events.
Every other event is a \resolvev. The exact position of termination (i.e., \resolvev events) is not specified in global types, so we do not constrain them. 
Reading from a location is defined analogously. 
The translation of $\G_1. \G_2$ models that there is a position $i$ such that, for every object \objname, 
the events described in $\histype(\G_1)$ are in the subtrace before $i$ and those in $\histype(\G_2)$ are in the subtrace after $i$.

\EK{The restriction is applied for every object, to ensure the following property:
If a trace is a model for the translation of a type $\Prcl$, then for each participating object (1) all events of this objects have the same order as specified in $\Prcl$ and (2) at the moment of the event, the corresponding FO formula holds.
The translation of, e.g., $\GInterOn{\ROLE{\objname_1}}{}{\ROLE{\objname_2}}{\methodname_2}.\GInterOn{\ROLE{\objname_1}}{}{\ROLE{\objname_3}}{\methodname_3}$
describes that ${\ROLE{\objname_2}}.{\methodname_2}$ is called before ${\ROLE{\objname_2}}.{\methodname_2}$, but does \emph{not} describe that the execution start in the same order. 
Thus, there are multiple possible event order satisfying this constraint, but from \emph{every local point of view} the differences between these traces are not visible.
}


\section{Analysis}\label{sec:type}
Verifying deadlock freedom requires a \emph{Points-To} analysis in addition to a type system.
Deadlock freedom is equivalent to cycle-freedom of causality graphs~\cite{MHPDead} in Active Objects.
The \emph{causality graph} of a global type $\Prcl$ is $\caus(\Prcl) = (V,E)$. 
Each node $\T \in V$ is a local type, 
and each edge $(\T_1,\T_2) \in E$ models that $\T_2$ must happen after $\T_1$.
\begin{definition}[Causality Graph]
Let $\Prcl$ be a well-formed global type. 
The nodes of its causality graph $\caus(\Prcl)$ 
are all partial local types derived from projecting $\Prcl$ on all endpoints. 
An edge $(\T_1,\T_2)$ is added if either (1) $\T_1 = \T.\T_2$ is a partial type for some $\T$
in some projection on some object 
or (2) $\T_1$ is the sending type and $\T_2$ the receiving type from the projection of a single calling type.
\end{definition}
Note that global types do not contain sufficient information to deduce all causality, 
e.g., the causality of \abs{get} statements cannot be deduced from a global type
because synchronizations on futures are specified over \emph{locations}.
We use a Points-To analysis for futures~\cite{MHPDead} instead. 
For generating a causality graph, 
we first derive a \emph{partial} causality graph from the global type, 
and then we apply the Points-To analysis during type checking for the graph completion by deducing the missing edges.
The Points-To analysis, defined below, determines which methods are responsible to resolve the futures in a given expression. 
\begin{definition}[Points-To] \label{def:p2}
The \emph{Points-To analysis} determines the set $\mathsf{p2}(\expr)$ of methods, which may have resolved the future stored in an input expression $\expr$.
We can express this using constraints, to integrate it into the type system:\footnotesize\vspace{-1mm}
\begin{align*}
   &\forall i\in\mathbb{N}.~ \config(i) \doteq \task{\objname', \future, \retvalue(\expr'),\lstore}~\task{\objname, \future', \methodname'(\var = \expr\xabs.\getABS; \statement''),\lstore''}~\config  \wedge \eval{\expr} = \future\rightarrow \\
   &\exists j\in\mathbb{N}.~ j < i \wedge \config(j) \doteq \task{\objname', \future, \methodname(\statement),\lstore'}~\config' \wedge \methodname \in \mathsf{p2}(\expr)
\end{align*}\normalsize
\end{definition}
Whenever a $\expr.\xabs{get}$-statement is checked against a type $\ReadAs{\expr}$, edges are added between the node of termination type of the methods which \expr can point to, and the node of the current type $\ReadAs{\expr}$.
\EK{Although Points-To is undecidable, well-scaling tools which safely overapproximate are available~\cite{MHP}.} 
%

\begin{definition}[Admissibility]
A causality graph is \emph{admissible} if
(1) every path is cycle-free and (2)
for every object \objname, and for any pair of receiving types of \objname, there exists a connecting path without an edge of the form $(\PutAs{\phi},\receiveTyped{\methodname}{\psi})$.
\end{definition}
\EK{
The graph on page~\pageref{graph} is admissible. With a non-admissible graph, methods may deadlock (violating (1)) or be executed in the wrong order (violating (2)).
}

\paragraph{Type System and Analysis.} 
The auxiliary \ADL-formula  
$\mathsf{post}(\actor.\methodname,\phi)$
models that the value in every future resolved by $\actor.\methodname$ satisfies $\phi$, 
while formula $\mathsf{Post}(\Prcl)$ represents the conjunction of all postconditions specified in $\Prcl$. 
Figure~\ref{fig:type} shows selected typing rules invoking the validity calculus~\cite{Din14} and Points-To analysis.

Before introducing the typing rules, 
we define $\mathsf{Roles}(\Prcl)$ as the set of objects in $\Prcl$, 
$\caus(\Prcl) + E$ as the set of edges of $\caus(\Prcl)$ and $E$ (i.e., $E$ is added into $\caus(\Prcl)$), 
$\textsf{term}(\methodname)$ as the set of $\downarrow$ nodes of method \methodname, 
and $\mathsf{node}(\statement)$ as the set of nodes referring to the types that have typed $\statement$.
We define three kinds of type judgments:
\paragraph{(I) The Type Judgment for \emph{Programs}.}
$\proves \prgm \hastype \Prcl$ checks $\prgm$ against global type $\Prcl$. 
The well-formedness of $\Prcl$ (Def.~\ref{def:well}) is ensured during type checking.
Rule \rulename{T-Main} checks that every endpoint in $\Prcl$ is implemented in $\prgm$, the main block makes the correct initializing call and checks each object against its object type.
The edges collected from the typing rules for objects are added to the partial causality graph $\caus(\Prcl)$ and the resulting graph is checked for admissibility.

\paragraph{(II) The Type Judgment for \emph{Objects}.}
$\Phi \proves \objlang \hastype \LPrcl \has E$
checks whether $\objlang$ is well-typed by $\LPrcl$ under a given $E$ with $\Phi$. 
$E$ is a set of causality edges and $\Phi$ is a set of \ADL formulas.
Rule \rulename{T-Object} projects $\LPrcl$ on all methods, checks each method $\methodname_i$ by $\project{\LPrcl}{}{\methodname_i}$ and collects all resulting edges.

\begin{samepage}
\begin{figure}[!thb]\centering
\scalebox{0.95}{%
$\TINFER{T-Main}{
\deduce{
\exists j \leq n.~\Prcl = \GInterOn{\xabs{main}}{}{\objname_j}{\methodname\obligation{\phi}}.\G \quad 
\forall i\leq n.~\mathsf{Post}(\Prcl) \proves \objlang_i \hastype \prop^{\ast}(\project{\Prcl}{}{\objname_i}) \has E_i
}{
\objlang_i = \xabs{object}~\objname_i \{\dots\} \quad \!
\mathsf{Roles}(\Prcl) = \{\objname_1,\dots,\objname_n\} \qquad \caus(\Prcl) + \bigcup_{i\leq n}E_i\text{ admissible}
}
}{
\proves \objlang_1\quad\dots\quad\objlang_n \quad \xabs{main}\{\objname_j!\methodname()\} \hastype \Prcl 
}
$}\vspace{2mm}
\scalebox{0.9}{%
$\TINFER{T-Object}{
    \deduce{
\forall i \leq n.~ \Phi,\phi_i,\xabs{skip} \proves \statement_i \hastype \T_i \has E_i \qquad E = \bigcup\nolimits_{i \leq n}E_i
    }{
\forall i \leq n.~ \project{\LPrcl}{\act}{\methodname_i} = \receiveTyped{\methodname_i}{\phi_i}.\T_i
    }
}{
\Phi \proves \xabs{object}~\objname\{\type_1~\methodname_1(\many{\type~\var})\{\statement_1\}  ~\dots ~
\type_n~\methodname_n(\many{\type~\var})\{\statement_n\} \quad \many{\type~\var=\expr}\} \hastype \LPrcl \has E
}
$}\vspace{2mm}
\scalebox{0.9}{%
$\TINFER{T-Return}{
\Phi \Rightarrow [\statement;\returnABS~\expr]\phi 
}{
\Phi, \statement \proves \returnABS~\expr \hastype \PutAs{\phi} \has E
}
\qquad
\TINFER{T-Call}{
\deduce{
\Phi \Rightarrow [\statement;\type~\var = \objname!\methodname(\many{\expr})]\phi
}{
\Phi, \statement;\type~\var = \objname!\methodname(\many{\expr}) \proves \statement' \hastype \T \has E
}
}{
\Phi, \statement \proves \type~\var = \objname!\methodname(\many{\expr});\statement' \hastype 
\sendTyped{\objname}{\var}{\methodname}{\phi}.\T \has E
}
$}\vspace{2mm}
\scalebox{0.9}{%
$\TINFER{T-Get}{
\deduce{
E = E' \cup \{(n,n') | \exists \methodname \in \mathsf{p2}(\expr).~ n \in \mathsf{term}(\methodname) \wedge n' \in \mathsf{node}(\statement;\expr.\getABS)\}
}{
\Phi, \statement;\type~\var = \expr.\getABS \proves \statement' \hastype \T \has E'
}
}{
\Phi, \statement \proves \type~\var = \expr.\getABS;\statement' \hastype 
\ReadAs{\expr}.\T \has E
}
$}
\caption{The selected typing rules.}
\label{fig:type}
\end{figure}
\paragraph{(III) The Type Judgment for \emph{Statements}.} 
$\Phi,\statement \proves \statement \hastype \T \has E$
checks whether $\statement$ is well-typed by $\T$ under a given $E$ with $\Phi,\statement$.
The environment $\statement$ are the statements type-checked so far. Whenever an \ADL formula is checked, a validity check is performed and \statement is added in the modality to consider the side-effects on the heap memory so far.
However, these are not recorded in $E$: The causality edges only record which method a \abs{get} statement synchronizes on.
Rule \rulename{T-Return} checks that after executing all the type-checked statements, the $\returnABS$ statement results in a state where $\phi$ holds.
Rule \rulename{T-Call} also checks the formula $\phi$ which describes the state when the call has to be executed.
Rule \rulename{T-Get} additionally executes the Points-To analysis and adds all the edges as described in the previous section.
\end{samepage}

\begin{theorem}[Deadlock Freedom and Protocol Adherence]\label{thm:main}
Let $\prgm$ be a program and $\Prcl$ be a global type.
If $\prgm$ is well-typed against $\Prcl$ then (1) $\prgm$ does not deadlock 
and (2) every generated trace from $\prgm$ satisfies $\histype(\Prcl)$: \small\[\proves \prgm \hastype \Prcl   \rightarrow \big( \forall \trace.~\prgm\!\Downarrow\!\trace \rightarrow \trace \models \histype(\Prcl) \big)\] \normalsize
\end{theorem}
\vspace{-5mm}

\section{Loops and Repetition}\label{sec:rep}
In this section we present the whole workflow of the previous section for \coreactor extended with repetition.
The language is extended with loops and the types with \emph{repetition} types $(\G)^{\ast}_{\phi}$ (resp. $(\T)^{\ast}_{\phi}$).
A repetition type resembles a Kleene-star and models the finite repetition of the type $\G$ (resp. $\T$). 
The formula $\phi$ is a loop invariant and has to be satisfied whenever a loop iteration starts or ends. 
\begin{definition}[Syntax with Repetition]
\[
\statement ::= \dots \sep \xabs{while}(\expr)\{\statement\}\qquad
\G ::= \dots \sep(\G)^{\ast}_{\phi}.\G\qquad
\T ::= \dots \sep (\T)^{\ast}_{\phi}\nnn\T
\]
\end{definition}

By syntactic restrictions, the local type $\LPrcl$ of an object cannot have the form $(\T)^\ast_\phi.\T'$, which forbids it to start with a loop.
The intuition behind this restriction is that
every loop has an invariant that an object must guarantee before 
executing the next iteration. 
If an object is not active before the loop, it cannot guarantee the invariant in the very beginning, thus repetition can start with the second action at the earliest.
Below give an example for using invariants.
\begin{example}\label{ex:rep}
Consider a big data analysis webtool with a client-side GUI \ROLE{U} and a server-side computational server \ROLE{S}. We model the following scenario: 
\begin{floatingfigure}[r]{0.25\textwidth}\hspace{-5mm}
\begin{tikzpicture}[scale=0.75, every node/.style={transform shape}]

    \node [copy shadow={draw=black!30,fill=black!30,shadow xshift=0.25ex,
        shadow yshift=-0.25ex},fill=white,draw=black,thick,font=\bfseries]
        at (1,4.5) {\ROLE{U}};

    \node [copy shadow={draw=black!30,fill=black!30,shadow xshift=0.25ex,
        shadow yshift=-0.25ex},fill=white,draw=black,thick,font=\bfseries]
        at (3,4.5) {\ROLE{S}};

    \draw[dotted] (1,1) -- (1,4.2);
    \draw[dotted] (3,1) -- (3,4.2);

    \draw[->] (0.25,4) -- (0.75,4);
    {\draw[->] (1.25,3.5) -- (2.75,3.25);}
    {\draw[->] (2.75,2.75) -- (1.25,2.5);}

    {\draw[->] (1.25,1.75) --  (2.75,1.5);}

    \filldraw[draw=black, fill=black!50] (0.75,3) rectangle (1.25,4);
    \node[rotate=90] at (1,3.5) {run};
    {\filldraw[draw=black, fill=black!50] (0.75,2.5) rectangle (1.25,1.75);}
    \node[rotate=90] at (1,2.125) {up};
    {\filldraw[draw=black, fill=black!50] (2.75,1) rectangle (3.25,3.25);}
    \node[rotate=90] at (3,2) {comp};

\draw[dashed] (0.5,2.85) -- (0.5,1.25) -- (3.5,1.25) -- (3.5,2.85) -- (0.5,2.85);
\end{tikzpicture}
\end{floatingfigure}
\ROLE{U} sends data to the computational server by calling \xabs{S.comp}. To stay responsive, \ROLE{U} ends its initial process. \ROLE{U} is called repeatedly on \abs{G.up} by the server to update the progress. 
Whenever \ROLE{U} is updated, the server also gets information by reading from the future of the last call to \abs{U.up}.
The sequence diagram to the right illustrates the protocol. 
During updating, \ROLE{U} must stay in a state expecting to receive updates from the server. It is therefore important to specify that field \abs{U.expect} is not \xabs{Nil}.
\footnotesize
\begin{align*}
 \GInterOn{\mainO}{}{U}{\mathsf{run}\obligation{\top}} \nnn
 \GInterOn{U}{}{S}{\xabs{comp}\obligation{\top,\top}} \nnn
 \Big(
 \GInterOn{S}{\xabs{x}}{U}{\xabs{up}\obligation{\top,\top}} \nnn
 \glGet{S}{\xabs{x}} 
 \Big)^{\ast}_{\ROLE{U}.\xabs{expect} \neq \xabs{Nil}}
 \nnn
\mathbf{end}
\end{align*}\normalsize
The invariant $\ROLE{U}.\xabs{expect} \neq \xabs{Nil}$ specifies the condition that field \abs{U.expect} is a non-empty list.
This is propagated during projection, which results in the following local type for \abs{U.up}
\[
\receiveTyped{\xabs{update}}{\self.\xabs{expect} \neq \xabs{Nil}}.\PutAs{\self.\xabs{expect} \neq \xabs{Nil}} 
\]
\noindent There is no repetition because \emph{being repeatedly called} is only visible for the whole object, not a single process. The type of \abs{S.cmp} however contains a repetition:
\[
\receiveTyped{\xabs{comp}}{\top}.\big(\sendTyped{\xabs{U}}{\xabs{x}}{\xabs{up}}{\top}.\ReadAs{\xabs{x}}\big)_{\exists l.~l \neq \xabs{Nil}}.\PutAs{\top}
\]
\end{example}

The workflow is the same as described above. We provide the projection, translation, propagation and typing rules as extension of the previous systems.
\begin{definition}[Projection Rules for Loops]\newline
The auxiliary predicate $\mathit{rcv}(\objname,\G)$ holds if $\objname$ is specified as being called in $\G$.
\small\begin{align*}
\project{ \big((\G)^{\ast}_{ \phi }.\G'\big) }{\act}{\objname} &= \casedin{
\PutAs{\act(\objname)}  .  (\T'')^{\ast}_{ \projectphi{\phi}{\objname}  } . \T'
&\begin{array}{l}
\text{if }\T \neq \skipT \!\wedge\! 
\phantom{\neg}\mathit{rcv}(\actor,\G) \!\wedge\! \act(\actor)\!\neq\! \bot \!\wedge\! \mathit{cs}(\phi)
\end{array}\\
(\T)^{\ast}_{ \projectphi{\phi}{\objname}  } . \T'
&\begin{array}{l}
\text{if }\T \neq \skipT \!\wedge\! 
\neg\mathit{rcv}(\objname,\G) \!\wedge\! \act(\objname)\!\neq\!\bot \!\wedge\! \mathit{cs}(\phi)
\end{array}\\
\T'&\text{if }\T = \skipT \!\wedge\! \mathit{cs}(\phi)
}\\ 
&\text{Where $\project{\G.\kend}{\act[\objname \mapsto \bot]}{\objname} = \T''$,$\project{\G}{\act}{\objname} = \T$ and $\project{\G'}{\act}{\objname} = \T'$}
\end{align*}\normalsize
\end{definition}
The auxiliary formula $\mathit{cs}(\phi)$ specifies that all weakenings of $\phi$ imply $\phi$. 
This is necessary to reject invariants that connect multiple heaps: e.g., this condition would reject $\xabs{G.i} \doteq \xabs{S.i}$, as it cannot be guaranteed by \abs{G} and \abs{S} separately. 
This condition, however, admits $\xabs{G.i} \doteq 1 \wedge \xabs{S.i} = 1$.
The first rule projects global types to object types.
The first case is applied if the object participates in the repetition of the inner type $\G$ \emph{by being repeatedly called}. 
The last active process must terminate first and the repeatedly called method must terminate within the repetition. The termination inside the loop is ensured by projecting the inner type with an appended $\kend$.
The second case is applied if the object participates in the repetition ($\T \neq \skipT$) by any other repeated action then being called ($\neg\mathit{rcv}(\objname,\G)$). 
Finally, the last case skips the repetition if the object does not participate in it.

The second rule projects object types to methods.
The rule distinguishes whether the whole process is inside the repetition or not. If the process is completely inside, the repetition is removed, as it is not visible to the method.

In presence of repetition, invariants have to be propagated inside the repeated, the previous, and the next types.
The following definition summarizes gives the rules for repetition, additionally to rule $\rulename{1}$ in Def.~\ref{def:prop}.
\begin{definition}[Rules for Propagation for Repetition] 
\small\begin{align*}
&\rulename{2}~\PutAs{\phi}.(\T)^{\ast}_{\psi} \rightsquigarrow \PutAs{\phi \wedge \psi}.(\T)^{\ast}_{\psi} \qquad
\rulename{3}~(\T)^{\ast}_{\psi}.\receiveTyped{\methodname'}{\phi} \rightsquigarrow (\T)^{\ast}_{\psi}.\receiveTyped{\methodname'}{\phi\wedge\psi} \\ 
&\rulename{4}~(\T)^{\ast}_{\phi}.(\T)^{\ast}_{\psi} \rightsquigarrow (\T)^{\ast}_{\phi\wedge\psi}.(\T)^{\ast}_{\psi} 
\quad\rulename{5}~(\receiveTyped{\methodname'}{\phi}.\T.\PutAs{\phi'})^{\ast}_{\psi} \rightsquigarrow (\receiveTyped{\methodname'}{\phi\wedge\psi}.\T.\PutAs{\phi'\wedge\psi})^{\ast}_{\psi}
\end{align*}\normalsize
\end{definition}
Since loop invariants have to hold \emph{before} the first repetition, rule $\rulename{2}$ ensures that the last process before a repetition satisfies the invariant when terminating.
Rule $\rulename{3}$ adds an invariant to the next process, as the invariant also holds \emph{after} the last repetition. Rule $\rulename{4}$ is another case of the first one, in case two repetitions are succeeding each other. 
Finally, rule $\rulename{5}$ adds the invariant to the processes inside the repetition. 
This rule enables the use of the invariant in the first method of the repetition and ensures that the last method reestablishes the invariant. 

For the translation into constraints, first-order constraints are not expressive enough. The Kleene star constraint resembles regular languages and we thus use \emph{monadic second order logic} (MSO) to capture repetition.
MSO extends first-order logic with a quantifier $\exists Y \subseteq Z$ which quantities over subsets of $Z$  and a $\in$ primitive to express membership of those sets.
The extension of relativization is straightforward~\cite{henkin}.
We now extend the semantics of types as constraints from Def.~\ref{def:msosem} to repetition:
\begin{definition}[Semantics of Repetition]
The semantics of repeated types uses a set of boundary indices $X$, between which the inner translation must. Also, the invariant has to hold at every boundary.
\small\begin{align*}
\histype\big((\G^{\ast}_{\phi})\big) ~=~&\exists X \subseteq \mathbb{N}.~\exists i,j \in X.~ \big(\forall k \in \mathbb{N}.~ i < k \leq j\big)\wedge \forall i \in X.~ \config(i) \models \phi\wedge \\
&\forall i,j \in X.~ \Big(\big(\forall k\in X.~ k \geq j \vee k \leq i\big) \Rightarrow \big(\histype(\G)\big)[n \in\mathbb{N} / i < n \leq j]\Big)
\end{align*}\normalsize
\end{definition}

\noindent The typing rule for repetition resembles invariant rules from Hoare calculi~\cite{hoare}: 
\[
\INFER{T-While}
{
    \deduce{
         \phi\wedge\mathsf{Post}(\Prcl),\xabs{skip} \proves \statement \hastype \T \has E' \qquad E = E' \cup E''
    }{
         \phi\wedge\mathsf{Post}(\Prcl),\xabs{skip} \proves \statement' \hastype \T' \has E'' \qquad \Phi \Rightarrow [\statement'']\phi \qquad \phi\wedge\mathsf{Post}(\Prcl)\Rightarrow [\statement]\phi
    }
}
{
    \Phi, \statement'' \proves \whileABS \expression\ \{\statement\}; \statement' \hastype (\T)^{\ast}_{\phi}\nnn\T' \has E
}
\]
The first premise continues the type checking of the program, in an environment where only the information in the invariant (and the global information in \postConditions, as defined in Section~\ref{sec:type}) is available.
The second and third premises check that the invariant holds initially and is preserved by the loop body.
The forth premise checks the loop body and the last premise combines the derived causality edges.
The extension of the causality graph is described in~\cite{thesis}.
\begin{corollary}
Theorem~\ref{thm:main} holds for the system with repetition.
\end{corollary}

\COMMENT{
\begin{example}
The code in Figure~\ref{fig:repfig} implements the behavior of the GUI and the Server from Example~\ref{ex:rep} and can be checked with the types from Example~\ref{ex:app2}.
\begin{figure}
\begin{abscode}
object U{
 List<Int> expect = Nil;
 Unit run(List<Int> l){ S!comp(l); expect = l;}
 Bool update(Int val, Int res){    
   return updateScreen(expect,val,res); //crashes if called in wrong state
 }
}
object S{
 Unit comp(List<Int> l){
    Int i = 0;
    while(i < length(l)){
      Int res = compute(l[i]); //dummy
      Fut<Bool> b = U!update(l[i],res);
      b.get; i++;
    }
 }
}
\end{abscode}
\caption{The code for Example~\ref{ex:rep}, with loop.}
\label{fig:repfig}
\end{figure}

\end{example}
}

\section{Branching}\label{sec:branch}
Active Object have multiple ways to communicate the choice how to continue the protocol and how an object reacts on it:
\begin{itemize}
\item
(1) The choice is communicated via method selection, i.e., each branch corresponds to a different method call.
\item 
(2) The choice is communicated via futures, i.e., other objects must react to the choice of an object by reading its future.
\item 
(3) The choice is communicated via the heap memory, i.e., processes must behave according to some condition for the memory.
\end{itemize}

We aim to stick with standard imperative statements and must regard the restriction that an \abs{if} statement can only choose between two branches, while a protocol may describe more than two.
In our analysis of branching, choice is communicated:
\begin{itemize}
\item
(1) method calls and condition on the passed data for new process running on other objects
\item 
(2) the condition on future for already running processes running on (possibly) other objects.
\item 
(3) via post-conditions to processes running later on the same object.
\end{itemize}
\begin{definition}[Syntax with Branching]
\begin{align*}
\statement &::= \dots \sep \xabs{if}(\expr)\statement~\xabs{else}~\statement~\xabs{fi}\qquad
\G ::= \dots \sep \objname\big\{\obligation{\phi_i}  ,   \big(\objname_{ij}\obligation{\phi_{ij}}\big)_{j\in J} ; \G_i \big\}_{i \in I}\\
\T &::= \dots \sep  \select \{\T_i\}_{i \in I} \sep \offer \{\objname.\methodname\obligation{\phi_i} ;  \T_i\}_{i\in I}
\end{align*}
\end{definition}
The global type $\objname\big\{\obligation{\phi_i}  ,   \big(\objname_{ij}\obligation{\phi_{ij}}\big)_{j\in J} ; \G_i \big\}_{i \in I}$ describes that $\objname$ chooses a branch $\G_i$.
The formulas $\phi_i$ are \emph{additional} postconditions for the choosing process. Other process can read the choice by reading this future. In $\objname_{ij}\obligation{\phi_{ij}}\big)_{j\in J}$,
we describe that the currently active process of $\objname_{ij}$ has the additional postcondition $\phi_{ij}$.
The local type $\select \{\T_i\}_{i \in I}$ is an active choice and $\offer \{\objname.\methodname\obligation{\phi_i} ;  \T_i\}_{i\in I}$ is a passive choice. The branch must be taken by reading the future from $\objname.\methodname$
and evaluating $\phi_i$. 

\begin{definition}[Projection Rules for Branching]
Given the $i$th branch $\obligation{\phi_i}, \big(\objname_{ij}\obligation{\phi_{ij}}\big)_{j\in J} ; \G_i$, we denote the updated \act function
with
\[\act_i =  \act[\objname \mapsto \act(\objname) \wedge \projectphi{\phi_i}{\objname}][\objname_{ij} \mapsto \act(\objname_{ij}) \wedge \projectphi{\phi_{ij}}{\objname_{ij}}]_{j\in J}\]
The auxiliary predicate $\mathsf{allAct}$ states that all mentioned objects and occur in all branches are active and $\mathsf{dist}$ states that a set of formulas does not overlap.
\begin{align*}
\mathsf{allAct} &= \act(\objname) \neq \bot \wedge \bigwedge\nolimits_{\substack{i\in I \\j \in J}}\act(\objname_{ij}) \neq \bot \wedge \forall i,i' \in I.~\forall j.~\objname_{ij} = \objname_{i'j}\\
\mathsf{dist}(\{\phi_1,\dots,\phi_n\}) &= \forall i,j<n.~i \neq j \rightarrow (\phi_i \wedge \phi_j \text{ is unsatisfiable})
\end{align*}
Figure~\ref{fig:branchproj} shows the projection rules for branching.
\end{definition}
The projection rule from global to object-local types has four cases: the first two are straightforward for the choosing process and the currently active reacting processes.
The third case handles objects which behave the same in all branches and the forth handles objects which are active in only one.
The projection on the passive choice moves the $\ReadAs{}$ type from its position after the choice in front of it: The global type has no explicit point where a process terminates, thus the read must be after the choice which adds the postcondition
to the choosing process. However the \abs{get} statement must be before the \abs{if} statement, which relies on the read value in the guard.

\begin{figure}[bth]
\small
\begin{align*}
\project{\Big(\objname\big\{\obligation{\phi_i}  ,&   \big(\objname_{ij}\obligation{\phi_{ij}}\big)_{j\in J} ; \G_i \big\}_{i \in I}\Big)}{\act}{\objname'}
=\\
&\casedin{
\select\{\T_i\}_{i\in I} & \text{if }\objname = \objname' \wedge \mathsf{allAct}
\wedge \project{\G_i}{\act_i}{\objname'} = \T_i\\
\offer\{\objname.\methodname\obligation{\projectphi{\phi_i}{\objname'}}; \T_i\} & \text{if } \objname_{ij} = \objname' \wedge \mathsf{allAct}
\wedge \project{\G_i}{\act_i}{\objname'} = \T_i\\
\T &\text{if }\act(\objname') = \bot \wedge \forall i.~\project{\G_i}{\act_i}{\objname'} = \T\\
\T &\text{if }\act(\objname') = \bot \wedge \exists i.~\project{\G_i}{\act_i}{\objname'} = \T \\
&\hspace{19.5mm}\wedge \forall j\neq i.~\project{\G_j}{\act_j}{\objname'} = \skipT
}
\end{align*}
\begin{align*}
\project{\select \{\T_i\}_{i \in I}&}{\methodname'}{\methodname} =\\ 
&\casedin{
\bigcup_{i\in I} \project{\T_i}{\methodname'}{\methodname} & \text{if }\methodname \neq \methodname'\\
\big\{\select \{\T_i\}_{i \in I}\big\}& \text{if }\methodname = \methodname' \wedge \forall i \in I.~\T_i\upharpoonright_{\methodname'}\methodname = \T_i'
}\\
\project{\offer \{\objname.\methodname\obligation{\phi_i} ;  \T_i\}_{i\in I}&}{\methodname'}{\methodname} =\\
&\casedin{
\bigcup_{i\in I} \project{\T_i}{\methodname'}{\methodname} & \text{if }\methodname \neq \methodname'\\
\big\{\ReadAs{\expr}.\offer \{\objname.\methodname\obligation{\phi_i} ;  \T_i\}_{i\in I} \big\} & \text{if }\methodname = \methodname' \wedge \mathsf{dist}((\phi_i')_{i_\in I})\wedge\\
&\phantom{\text{if }}\forall i \in I.~\T_i\upharpoonright_{\methodname'}\methodname = \T_i' = \ReadAs{\expr}.\T_i''
}
\end{align*}
\caption{Projection Rules for Branching}
\label{fig:branchproj}
\end{figure}

\begin{definition}[Translation into MSO for Branching]
For the translation into MSO constraints, we use the auxiliary predicate $\firstput(i,\objname)$ that states that the $i$th position in the trace refers to the first resolving event from \objname
and the auxiliary predicate $\lastput(i,\objname.\methodname)$ that states that the $i$th position in the trace refers to the last resolving event of $\objname.\methodname$.
\begin{align*}
\firstput(i,\objname) =& \forall j.~ \big(\exists\future.~\exists\methodname.~\exists \expr.~ \ev(j) \doteq \resolvev(\objname,\future,\methodname,\expr)\big) \rightarrow i \leq j\\
\lastput(i,\objname.\methodname) =& \forall j.~ \big(\exists\future.~\exists \expr.~ \ev(j) \doteq \resolvev(\objname,\future,\methodname,\expr)\big) \rightarrow i \geq j\\
\end{align*}
Additionally to the translation of the branches, it encodes that the choosing process terminates before any process that relies on the communication of its choice via the return value.
The rules are as follows:\scriptsize
\begin{align*}
&\histype\Big(\objname\big\{\obligation{\phi_i}  ,   \big(\objname_{ij}\obligation{\phi_{ij}}\big)_{j\in J} ; \G_i \big\}_{i \in I}\Big) =\\
\bigvee\nolimits_{i \in I}\Big(\histype(\G_i) \wedge \exists k.~\firstput(k,\actor) \wedge &\bigwedge\nolimits_{j \in J}\big(\exists k_j.~\firstput(k_{j},\objname_{ij}) \wedge k \geq k_j  \wedge \sigma(h)[k_{j}] \models \phi_{ij}\big)\Big)
\end{align*}\vspace{-4mm}
\begin{align*}
\histype(\select \{\T_i\} ) &= \bigvee_{i}\histype(\T_i) \\
\histype(\offer \{\actor.\method\obligation{\phi_i};\T_i\}) &=
\bigvee_{i}\big(\exists j\in\mathbb{N}.~\lastput(j,\mathsf{p.m}) \wedge \sigma(j) \models \phi_i \wedge \histype(\T_i)\big)
\end{align*}\normalsize
\end{definition}
In the following we present the rules for branching.
The typing rules split the branches into two disjoint sets and shows that the guard of the \ifABS statement together with the added choice-conditions of the branch selects the correct continuation of the type.
Once the sets of branches are singletons, the choice operators can be removed. 
\begin{definition}[Typing Rules]
\[
\INFER{T-Offer}{
\begin{array}{c}
    I = I_1 \cup I_2 \qquad I_1 \cap I_2 = \emptyset \qquad E = E_1 \cup E_2\\
    \forall i \in I_1.~\Phi\wedge\mathsf{post}(\actor.\methodname, \phi_i) \Rightarrow \expression)\\
    \forall i \in I_2.~\Phi\wedge\mathsf{post}(\actor.\methodname, \phi_i) \Rightarrow  \neg\expression)\\
\Phi;\expression;\bigvee_{i\in I_1}\mathsf{post}(\actor.\methodname, \phi_i),\statement \proves \statement';\statement''' \hastype \offer\{\objname.\methodname\obligation{\phi_i};\T_i\}_{i\in I_1} \has E_1 \\
\Phi;\neg\expression;\bigvee_{i\in I_2}\mathsf{post}(\actor.\methodname, \phi_i),\statement \proves \statement';\statement''' \hastype \offer\{\objname.\methodname\obligation{\phi_i};\T_i\}_{i\in I_2} \has E_2
 \end{array}
}{
\Phi,\statement \proves \ifABS \expression\ \thenABS \statement'\ \elseABS \statement'' \ \fiABS\ ;\statement''' \hastype \offer\{\objname.\methodname\obligation{\phi_i};\T_i\}_{i\in I} \has E
}
\]
\[
\INFER{T-Offer-Single}
{
    \Phi,\statement' \proves \statement \hastype \T \has E
}
{
    \Phi,\statement' \proves \statement \hastype \offer\{\actor. \methodname : \phi; \T\} \has E
}
\INFER{T-Select-Single}
{
    \Phi,\statement' \proves \statement \hastype \T \has E
}
{
    \Phi,\statement' \proves \statement \hastype \select\{\T\} \has E
}
\]
\[
\TINFER{T-Select}
{
\deduce{
\Phi;\expression,\statement \proves \statement';\statement''' \hastype \select\{\T_i\}_{i\in I_1} \has E_1 \qquad
\Phi;\neg\expression,\statement \proves \statement'';\statement''' \hastype \select\{\T_i\}_{i\in I_2} \has E_2
}{
\deduce{
}{
I = I_1 \cup I_2 \qquad I_1 \cap I_2 = \emptyset \qquad E = E_1 \cup E_2
}
}
}{
\Phi,\statement \proves \ifABS \expression\ \thenABS \statement'\ \elseABS \statement'' \ \fiABS\ ;\statement''' \hastype \select\{\T_i\}_{i\in I} \has E
}
\]
\end{definition}
The extension of the causality graph is described in~\cite{thesis}.

We use the following example to illustrate how we handle branching.
\begin{example}\label{ex:app}
Consider the scenario: Client \ROLE{\objname_1} wants to access data on server \ROLE{\objname_2} and sends its login data by calling method $\methodfun{acc}$. Then \ROLE{\objname_2} decides. 
If the login data is invalid, 
\ROLE{\objname_2} logs the 
denied access by calling logging server \ROLE{S} and returns $-1$ to \ROLE{\objname_1}; 
if the access succeeds, 
it returns the data, a value $>0$, to $\actor_2$. 
\ROLE{\objname_1} reacts on the return value and returns a boolean indicating whether the access was successful.
This is formalized by the following type:\small
\[
\GInterOn{\mainO}{}{\objname_1}{\methodfun{start}}\nnn\GInterOn{\objname_1}{\xabs{x}}{\objname_2}{\methodfun{acc}}.\ROLE{\objname_2}\casedsclosed{
     \obligation{\mathbf{result} \doteq -1} ~ \ROLE{\objname_1}\obligation{\neg\mathbf{result}} ;~ \glGet{\objname_1}{\xabs{x}}.\GInterOn{\objname_2}{}{S}{\methodfun{log}}.\kend\! \\
     \obligation{\mathbf{result}>0} ~ \ROLE{\objname_1}\obligation{\mathbf{result}} ;~\glGet{\objname_1}{\xabs{x}}.\kend
}
\]\normalsize
The local type for $\objname_1.\methodfun{start}$ is the following. Note that the $\ReadAs{}$ type is now before the branching.
\small\begin{align*}
\receiveUntyped{\methodfun{start}}.\sendUntyped{\ROLE{\objname_2}}{\xabs{x}}{\methodfun{acc}}.\ReadAs{\xabs{x}}.\offer\casedsclosed{\objname_2.\methodfun{acc}\obligation{\mathbf{result} \doteq -1}~;~\PutAs{\neg\mathbf{result}} \\ 
                                                                \objname_2.\methodfun{acc}\obligation{\mathbf{result} > 0} ~;~\PutAs{\mathbf{result}}}
\end{align*}
\normalsize
\end{example}

\section{Conclusion and Related Work}\label{sec:conc}
In this paper we generalize MPST for Active Objects to a two-phase analysis that handles
protocols  where information is not only transmitted \EK{between objects }via asynchronous method calls
but also \EK{inside the object} through the heap memory of Active Objects. 
Additionally, we provide a model-theoretic semantics for MPST, which allows us to give a declarative definition of protocol adherence and integrate further static analyses. 
These analyses are used to reason about method order and future synchronization \EK{within a type system.} 


\subsection{Discussion}
\paragraph{Decidability and Types for Validation.}
The judgment $\proves \prgm \hastype \Prcl$ is undecidable if
the validity of the FO logic used for specifying side-effects is undecidable.
A developer can choose an FOL fragment with decidable validity to trade off expressiveness against analyzability, 
e.g., if the developer 
chooses a more restricted fragment, 
which may limit the expressiveness of the specification, 
then the validity of the FO logic used for specifying side-effects may become decidable.

When using an undecidable FOL fragment, our approach can be used as a \emph{validation} tool 
to check whether the implemented (sub-)system will be behaving as expected.
Our approach can be integrated into the development process similarly as invariant-based approaches, and 
applies techniques proposed by MPST to connect global and local views of concurrent programs,
a notoriously difficult problem when using contracts and invariants~\cite{Din14}.

\paragraph{Protocol Adherence.}
Current work on MPST defines protocol adherence as a fidelity theorem,
which states that every sequence of interactions in a session follows the scenario declared in MPST~\cite{Honda1} as follows: 
An operational semantics for types is defined and it is shown that the semantics of the language is a refinement of the semantics of the types. 
Similarly, behavioral contracts~\cite{Castagna} define protocol adherence by \emph{compliance}, which compares the interaction of contracts.
These are \emph{operational} approaches to specification. 
We define protocol adherence from a \emph{declarative} perspective by requiring a logical \emph{property} to hold for all traces of a well-typed program.
A declarative specification can be analyzed with tools for logical specification, 
and can enable easier integration of other static analysis tools (e.g., to consider state), because they are only required to have a logical characterization. 
\vspace{-2mm}
\subsection{Related Work}
This work extends our previous system for Active Objects~\cite{icfem}, which could not specify and verify state, required an additional 
verification step for the scheduler and explicit termination points within the global type.
\paragraph{Actors and Objects.}
Crafa and Padovani~\cite{crafa,arx} investigate behavioral types for the object-oriented join calculus with typestate, a concurrency model similar to actors. 
Gay et al.~\cite{oo2} model channels as objects, integrating MPST with classes; 
Dezani-Ciancaglini et al.~\cite{oo1} use MPST in the object-oriented language \texttt{MOOSE}, 
where types describe communication through shared channels. 
We ensure deadlock freedom 
similarly to Giachino et al.~\cite{Giachino20162,Giachino2016}, who ensure deadlock freedom by inferring behavioral \emph{contracts} and applying a cycle detection algorithm; 
however, they do not consider protocol adherence. 

\paragraph{State and Contracts.}
Bocchi et al.~\cite{Logic,contract1,Bocchi} develop a MPST discipline with assertions for endpoint state.
The work considers neither objects nor heap memory. 
The specifications use \emph{global values} in global types and require complex checks for \emph{history-sensitivity} and \emph{temporal-sensitivity} to ensure that an endpoint proves its obligations.
We evade this by specifying inherently class-local memory \emph{locations}. 
They explicitly track values over several endpoints, while 
we implicitly do so by equations over locations.
In a stateless setting, 
Toninho and Yoshida use dependent MPST~\cite{TONINHO201761} to reason about 
passed data.
\vspace{-2mm}
\paragraph{Logics.}
Session types as formulas have been examined by Caires et al.~\cite{Caires} and Carbone et al.~\cite{arbiter} for intuitionistic and linear logics as types-as-proposition for the $\pi$-calculus. 
\EK{Our work uses logic not for a \emph{proof-theoretic} types-as-proposition theorem, but to use a \emph{model-theoretic} notion of protocol adherence and to integrate static analysis and dynamic logic.
Lange and Yoshida~\cite{Lange} also characterize session types as formulas, but their characterization characterizes the \emph{subtyping} relation, not the execution traces as in our work.}
\subsubsection*{Acknowledgments} This work is partially supported by \texttt{FormbaR}, part of the Innovation Alliance between TU Darmstadt and Deutsche Bahn AG.

\clearpage
\bibliographystyle{abbrv}
\bibliography{bibliography}
\clearpage
\appendix

\section{Full Definitions}\label{sec:def}
\begin{definition}[Weakening]
Weakening is defined as \[\projectphi{\phi}{\objname} =  \underbrace{\exists~\type_1~v_1,\dots,\exists~\type_n~v_n.}_{\{v_1,\dots,v_n\} = \mathsf{free}(\phi)}\widehat\phi_{\objname}\]
where the set of all free variables in $\phi$ is denoted with $\mathsf{free}(\phi)$ 
and $\widehat{\cdot}_{~\objname}$ 
is defined as:\\
\begin{align*}
\widehat{\exists v.~\phi}_{\objname} &= \exists v.~\widehat{\phi}_{\objname}\quad
\reallywidehat{p(t,..,t_n)}_{\objname} = p(\widehat{t}_{\objname},..,\widehat{t_n}_\objname)\quad
\widehat{\neg\phi}_{\objname} = \neg\widehat{\phi}_{\objname}\quad
\widehat{\phi\vee\psi}_{\objname} = \widehat{\phi}_{\objname} \vee \widehat{\psi}_{\objname}\\
\widehat{\logicfunction}_{\objname} &= \cased{
\logicfunction & \text{ if $\logicfunction$ is a function symbol of $\objname$}\\
v_\logicfunction & \text{ otherwise, where $v_\logicfunction$ is a fresh logical variable with the type of $\logicfunction$}}
\end{align*}
\end{definition}

\begin{example}[Weakening]
Let $\field$ be a field, \objname an object and \abs{i} the parameter of some method in class $\objname$. 
Consider $\phi = \objname.\field > 0 \wedge \text{\abs{i}} > \objname.\field$.
The formula $\phi$ is a $\objname$-formula, as $\phi = \projectphi{\phi}{\objname}$.
The weakening for some object $\objname'$ is $\projectphi{\phi}{\objname'} = \exists \text{\abs{Int}}~a.a > 0 \wedge \text{\abs{i}} > a$. 
The following (valid) $\objname'$-formula describes that if $\projectphi{\phi}{\objname'}$ holds in some state, then after executing $\text{\abs{j = i*2;}}$ the program reaches a state, where the variable $j$ contains a positive value:
\[\exists \text{\abs{Int}}~a.~\big(a > 0 \wedge \text{\abs{i}} > a \Rightarrow [\text{\abs{j = i*2;}}] \text{\abs{j}} > 0\big)\]
While $\objname'$ can not reason about the value of $\objname.\field$, the weakening allows to carry over the information that the parameter is larger than 1.
\end{example}

\begin{definition}[Relativation] \label{def:relatvie}
Let $\phi$ be a MSO constraint with a free variable $x$ of type $\type$ and $\psi$ a MSO constraint.
We denote the relativization of $\psi$ with $\phi$ by $\psi[ x\in Z / \phi ]$. 
For all quantifiers of type $\type$ in formula $\psi$, relativization adds $\phi(x)$ as restrictions into $\psi$.
The construction is defined with the following rules:\small
\begin{align*}
(\exists y \in Z.\psi)[ x\in Z / \phi ] =& \exists y\in Z. \phi(y) \wedge \psi [ x\in Z / \phi ]\\
(\exists Y \subseteq Z.\psi)[ x\in Z / \phi ] =& \exists Y \subseteq Z.~\big((\forall y\in Y. \phi(y)) \wedge \psi [ x\in Z / \phi ]\big)\\
(\phi \wedge \psi)[ x\in Z / \phi ] =& \phi[ x\in Z / \phi ] \wedge \psi[ x\in Z / \phi ]\\
(\neg \phi)[ x\in Z / \phi ] =& \neg(\phi[ x\in Z / \phi ])\\
(p(t_1,\dots,t_n))[ x\in Z / \phi ] =& p(t_1[ x\in Z / \phi ],\dots,t_n[x\in Z / \phi ])\\
(f(t_1,\dots,t_n))[ x\in Z / \phi ] =& f(t_1[ x\in Z / \phi ],\dots,t_n[x\in Z / \phi ])
\end{align*}
\normalsize
\end{definition}

\begin{example}[Relativation]
Consider a graph $(V,E,\mathsf{c})$ with one predicate \textsf{c} over its nodes. I.e., at every node $n$, the predicate $\mathsf{c}(n)$ either holds or not.
The formula  $\psi = \forall n \in V.~\mathsf{c}(n)$ expresses that \textsf{c} holds everywhere:
The following formula expresses that $x$ has an out-degree of at most 1:
\[\phi(x) = \forall y,z \in V.~E(x,y) \wedge E(x,z) \Rightarrow z \doteq y\]
The following formula restricts $\psi$ on the subgraph described by $\phi$, i.e. it expresses that at all nodes with an out-degree of at most 1, \textsf{c} holds:
\begin{align*}
\psi[x \in V\setminus\phi] = 
\forall n \in V.~\big(\forall y,z \in V.~E(n,y) \wedge E(n,z) \Rightarrow z \doteq y\big) \Rightarrow \mathsf{c}(n)
\end{align*}
\end{example}

\subsection*{Auxiliary Predicates}
\[\isput(i) = \exists \ROLE{A}.~ \exists f.~\exists \method.~\exists e \in D.~\ev(i) \doteq \resolvrev(\ROLE{A},f,e)\]

\subsection*{Projection on Objects}\footnotesize
\begin{align*}
\project{\GInterOn{\mathbf{main}}{}{\objname_2}{\methodname\obligation{\phi}}.\G}{\act}{\objname_1} &=
\receiveTyped{\methodname}{\projectphi{\phi}{\objname_2}} . (\project{\G}{\act[\objname_2 \mapsto \psi]}{\objname_1}) \text{ if $\act = \act_\bot \wedge \objname_2 = \objname_1$}\\ 
\project{\GInterOn{\mathbf{main}}{}{\objname_2}{\methodname\obligation{\phi}}.\G}{\act}{\objname_1} &=
\skipT. (\project{\G}{\act[\objname_2 \mapsto \psi]}{\objname_1}) \text{ if $\act = \act_\bot \wedge \objname_2 \neq \objname_1$}\\ 
\project{\GInterOn{\objname_1}{\var}{\objname_2}{\methodname\obligation{\phi,\psi}}. \G}{\act}{\objname_1} 
&=
  \sendTyped{\objname_2}{\var}{\methodname}{\phi} . (\project{\G}{\act[\objname_2 \mapsto \psi]}{\objname_1}) \text{ if } \act(\objname_1) \neq \bot \wedge \phi= \projectphi{\phi}{\objname_1}\\
\project{\GInterOn{\objname_1}{\var}{\objname_2}{\methodname\obligation{\phi,\psi}}. \G}{\act}{\objname_2} 
&=
  \casedin{
  \receiveTyped{\methodname}{\projectphi{\phi}{\objname_2}} . (\project{\G}{\act[\objname_2 \mapsto \psi]}{\objname_1})                           & \text{ if } \act(\objname_2) = \bot \\
  \PutAs{\act(\objname_2)}.\receiveTyped{\methodname}{\projectphi{\phi}{\objname_2}}. (\project{\G}{\act[\objname_2 \mapsto \psi]}{\objname_1})   & \text{ if } \act(\objname_2) \neq \bot\\
  }\\  
\project{\GInterOn{\objname_1}{\var}{\objname_2}{\methodname\obligation{\phi,\psi}}. \G}{\act}{\objname} 
&=
  \skipT. (\project{\G}{\act[\objname_2 \mapsto \psi]}{\objname}) \text{ if } \objname_2 \neq \objname \neq \objname_1 \\
  \project{(\glGet{\objname_1}{\expr}.\G')}{\act}{\objname} 
  &=
  \casedin{
  \ReadAs{\expr} . (\project{\G}{\act[\objname_2 \mapsto \psi]}{\objname})
  & \text{ if }\act(\ROLE{\objname_1}) \neq \bot  \wedge \ROLE{\objname_1} = \objname\\
  \skipT. (\project{\G}{\act[\objname_2 \mapsto \psi]}{\objname})
  & \text{ if } \ROLE{\objname_1} \neq \objname
  }\\
  \project{\kend}{\act}{\objname} 
  &=
\casedin{
  \PutAs{\act(\objname)}.\kend
  & \text{ if } \act(\objname) \neq \bot\\
  \kend
  & \text{ if } \act(\objname) = \bot\\
}
\end{align*}
\begin{align*}
\project{ \big((\G&)^{\ast}_{ \phi }.\G'\big) }{\act}{\objname} =\\ 
&\casedin{
\PutAs{\act(\objname)}  .  (\T'')^{\ast}_{ \projectphi{\phi}{\objname}  } . \T'
&\begin{array}{l}
\text{if }\T \neq \skipT \wedge 
\mathit{rcv}(\actor,\G) \wedge \act(\actor)\neq \bot \wedge \mathit{cs}(\phi)
\end{array}\\
(\T)^{\ast}_{ \projectphi{\phi}{\objname}  } . \T'
&\begin{array}{l}
\text{if }\T \neq \skipT \wedge 
\neg\mathit{rcv}(\objname,\G) \wedge \act(\objname)\neq \bot \wedge \mathit{cs}(\phi)
\end{array}\\
\T'&\text{if }\T = \skipT \wedge \mathit{cs}(\phi)
}\\ 
&\text{Where $\project{\G.\kend}{\act[\objname \mapsto \bot]}{\objname} = \T''$,$\project{\G}{\act}{\objname} = \T$ and $\project{\G'}{\act}{\objname} = \T'$}
\end{align*}
\[\mathit{cs}(\phi) = (\bigwedge_{\objname\in\mathit{objects}(\phi)}{\projectphi{\phi}{\objname}}) \rightarrow \phi\]


\subsection*{Projection on Methods}
\noindent\small\begin{minipage}{0.49\textwidth}
\[\receiveTyped{\methodname}{\phi}\upharpoonright_{\methodname'}\methodname = \casedin{ \{(\receiveTyped{\methodname}{\phi}, \methodname)\} & \text{ if }\methodname' = \bot \\ \{(\skipT,\methodname')\} & \text{ otherwise}}\]
\end{minipage}
\begin{minipage}{0.49\textwidth}
\[\PutAs{\phi}\upharpoonright_{\methodname'}\methodname = \casedin{\{ (\PutAs{\phi}, \bot)\} & \text{ if }\methodname = \methodname' \\ \{(\skipT,\bot)\} & \text{ otherwise}}\]
\end{minipage}
\begin{align*}
\ReadAs{e}\upharpoonright_{\methodname'}\methodname &= \casedin{ \{(\ReadAs{e}, \methodname')\} & \text{ if }\methodname = \methodname' \\ \{(\skipT,\methodname')\} & \text{ otherwise}}\\
\sendTyped{\objname_1}{v}{\methodname''}{\phi}\upharpoonright_{\methodname'}\methodname &= \casedin{ \{(\sendTyped{\objname_1}{v}{\methodname''}{\phi}, \methodname')\} & \text{ if }\methodname = \methodname'  \\ \{(\skipT,\methodname')\} & \text{ otherwise}}\\
\skipT \upharpoonright_{\methodname'}\methodname                &= \{(\skipT, \methodname')\} \qquad\qquad\kend \upharpoonright_{\methodname'}\methodname  = \{(\skipT, \methodname')\} \text{if }\methodname' = \bot\\
(\T_1.\T_2)\upharpoonright_{\methodname'}\methodname                &= 
\begin{array}{ll}
    &\{(\T_1'.\T_2',\methodname''') \sep (\T_1',\methodname'') \in \T_1\upharpoonright_{\methodname'} \wedge \methodname'' \neq \bot \wedge (\T_2',\methodname''') \in \T_2\upharpoonright_{\methodname''}\}\\
    \cup&\{(\T_2',\methodname''') \sep  (\T_1',\bot) \in \T_1\upharpoonright_{\methodname'}  \wedge (\T_2',\methodname''') \in \T_2\upharpoonright_{\bot}\}
\end{array}\\
(\T)^{\ast}_{\phi} \upharpoonright_{\methodname''}\methodname'           &= 
\casedin{ 
    \big\{\big((\T \upharpoonright_{\methodname''}\methodname')^{\ast}_{\phi}, \methodname''\big)\big\} & \text{ if } \methodname'' = \methodname' \wedge \T \upharpoonright_{\methodname''}\methodname' \neq \{(\skipT,\methodname''')\}\\ 
    \T \upharpoonright_{(methodname''}\methodname' & \text{ if } \methodname'' \neq \methodname' \wedge \T \upharpoonright_{\methodname''}\methodname' \neq \{(\skipT,\methodname''')\}\\ 
    \{(\skipT,\methodname'')\} &\T \upharpoonright_{\methodname''}\methodname' = \{(\skipT,\methodname''')\}
}
\end{align*}\normalsize

\subsection*{Translation into Constraints}
\small
\begin{align*}
\histype(&\GInterOn{\mainO}{}{\ROLE{\objname_2}}{\methodname \obligation{\psi}}.\G) ~=~\exists j,k .~\exists f.~\exists \expr'.~\event(j) \!\doteq\! \invocrev(\objname_2,f,\methodname) \wedge \config(j) \!\models\! \projectphi{\phi}{\objname_2} \wedge\\
&\event(k) \!\doteq\! \resolvev(\objname_2,f,\expr') \wedge \config(k) \!\models\! \psi  \wedge \forall l . l\!\neq\! j \wedge l\!\neq\! k \Rightarrow \isput(l) \wedge \histype(\G)\\
\histype(&\GInterOn{\ROLE{\objname_1}}{\var}{\ROLE{\objname_2}}{\methodname \obligation{\phi,\psi}} ) ~=~\exists i,j,k .~\exists f.~\exists \expr, \expr'.~\\
&\hspace{2mm}\event(i) \doteq \invocev(\objname_1,\objname_2,f,\methodname,\expr) \wedge \config(i) \models \phi \wedge \event(j) \doteq \invocrev(\objname_2,f,\methodname) \wedge \config(j) \models \projectphi{\phi}{\objname_2} \wedge\\
&\hspace{2mm}\event(k) \doteq \resolvev(\objname_2,f,\expr') \wedge \config(k) \models \psi \wedge \config(i) \models \objname_1.\xabs{x} \doteq f \wedge \forall l .~l\neq i \wedge l\neq j \wedge l\neq k \Rightarrow \isput(l) \\
\histype(&\GInterOn{\ROLE{\objname_1}}{}{\ROLE{\objname_2}}{\methodname \obligation{\phi,\psi}} ) ~=~\exists i,j,k .~\exists f.~\exists \expr, \expr'.~\\
&\hspace{2mm}\event(i) \doteq \invocev(\objname_1,\objname_2,f,\methodname,\expr) \wedge \config(i) \models \phi \wedge \event(j) \doteq \invocrev(\objname_2,f,\methodname) \wedge \config(j) \models \projectphi{\phi}{\objname_2} \wedge\\
&\hspace{2mm}\event(k) \doteq \resolvev(\objname_2,f,\expr') \wedge \config(k) \models \psi \wedge \wedge \forall l .~l\neq i \wedge l\neq j \wedge l\neq k \Rightarrow \isput(l) \\
\histype(&\glGet{\objname}{\expr})~=~ \exists i.~\exists f.~\exists \expr'.~\exists \objname'.~\event(i) \doteq \resolvrev(\objname,\objname',f,\expr') \wedge \config(i) \models \expr \doteq f \wedge \forall l.~l\neq i \Rightarrow \isput(l) \\
\histype(&\G_1 . \G_2) =\bigwedge_{\objname}\big(\exists i\in\mathbb{N}.~\histype(\G_1)[j\in\mathbb{N}/\actC(j,\objname) \Rightarrow j < i] \wedge \histype(\G_2)[j\in\mathbb{N}/\actC(j,\objname) \Rightarrow j \geq i]\big)\\
\histype\big(&(\G^{\ast}_{\phi})\big) ~=~\exists X \subseteq \mathbb{N}.~\exists i,j \in X.~ \big(\forall k \in \mathbb{N}.~ i < k \leq j\big)\wedge \forall i \in X.~ \config(i) \models \phi\wedge \\
&\forall i,j \in X.~ \Big(\big(\forall k\in X.~ k \geq j \vee k \leq i\big) \Rightarrow \big(\histype(\G)\big)[n \in\mathbb{N} / i < n \leq j]\Big) \\
\histype(&\kend) = \mathbf{true}\\
\end{align*}
The local translation for some object \objname is:
\begin{align*}
\histype(\T_1 \nnn \T_2)  ~&=~ \exists i \in \mathbb{N}.~\histype(\T_1)[n\in\mathbb{N}/ n< i] \wedge \histype(\T_2)[n \in \mathbb{N}/ n \geq i]\\
\histype(\receiveTyped{\methodname}{\phi}) &= \exists i.~ \forall j.~ i = j \wedge \exists f.~\invocrev(\objname,f,\methodname) \wedge \config(i) \models \projectphi{\phi}{\objname}\\
\histype(\sendTyped{\objname'}{\var}{\methodname}{\phi}) &= \exists i.~ \forall j.~ i = j \wedge \exists f,\expr.~\invocev(\objname,\objname',f,\methodname,\expr) \wedge \config(i) \models \projectphi{\phi}{\objname} \wedge \config(i) \models (\xabs{x}\doteq f)\\
\histype(\sendTyped{\objname'}{}{\methodname}{\phi}) &= \exists i.~ \forall j.~ i = j \wedge \exists f,\expr.~\invocev(\objname,\objname',f,\methodname,\expr) \wedge \config(i) \models \projectphi{\phi}{\objname}\\
\histype(\PutAs{\phi}) &= \exists i.~ \forall j.~ i = j \wedge \exists f,\expr.~ \resolvev(\objname,f,\expr) \wedge \config(i) \models \projectphi{\phi}{\objname}\\
\histype(\ReadAs{\expr}) &= \exists i.~ \forall j.~ i = j \wedge \exists f.~\resolvrev(\objname,\objname',f,\expr') \wedge \expr = f\\
\histype(\T^{\ast}_{\phi}) ~=~&\exists X \subseteq \mathbb{N}.~\exists i,j \in X.~ \big(\forall k \in \mathbb{N}.~ i < k \leq j\big)\wedge \forall i \in X.~ \config(i) \models \phi\wedge \\
&\forall i,j \in X.~ \big(\forall k\in X.~ k \geq j \vee k \leq i\big) \rightarrow \big(\histype(\T)\big)[x \in\mathbb{N} / i < x \leq j] \\
\histype(\kend) &= \mathbf{true}\\
\end{align*}
\subsection*{Typing Rules}
\small
$\TINFER{T-Main}{
\deduce{
\exists j \leq n.~\Prcl = \GInterOn{\xabs{main}}{}{\objname_j}{\methodname\obligation{\phi}}.\G \quad 
\forall i\leq n.~\mathsf{Post}(\Prcl) \proves \objlang_i \hastype \prop^{\ast}(\project{\Prcl}{}{\objname_i}) \has E_i
}{
\objlang_i = \xabs{object}~\objname_i \{\dots\} \quad i\!
\mathsf{Roles}(\Prcl) = \{\objname_1,\dots,\objname_n\} \qquad \caus(\Prcl) + \bigcup_{i\leq n}E_i\text{ admissible}
}
}{
\proves \objlang_1\quad\dots\quad\objlang_n \quad \xabs{main}\{\objname_j!\methodname()\} \hastype \Prcl 
}
$
\vspace{2mm}
$\TINFER{T-Object}{
    \deduce{
\forall i \leq n.~ \Phi,\phi_i,\xabs{skip} \proves \statement_i \hastype \T_i \has E_i \qquad E = \bigcup\nolimits_{i \leq n}E_i
    }{
\forall i \leq n.~ \project{\LPrcl}{\act}{\methodname_i} = \receiveTyped{\methodname_i}{\phi_i}.\T_i
    }
}{
\Phi \proves \xabs{object}~\objname\{\type_1~\methodname_1(\many{\type~\var})\{\statement_1\}  ~\dots ~
\type_n~\methodname_n(\many{\type~\var})\{\statement_n\} \quad \many{\type~\var=\expr}\} \hastype \LPrcl \has E
}
$

\vspace{2mm}
$\TINFER{T-Return}{
\Phi \Rightarrow [\statement;\returnABS~\expr]\phi 
}{
\Phi, \statement \proves \returnABS~\expr \hastype \PutAs{\phi} \has E
}
\qquad
\TINFER{T-Call}{
\deduce{
\Phi \Rightarrow [\statement;\type~\var = \objname!\methodname(\many{\expr})]\phi
}{
\Phi, \statement;\type~\var = \objname!\methodname(\many{\expr}) \proves \statement' \hastype \T \has E
}
}{
\Phi, \statement \proves \type~\var = \objname!\methodname(\many{\expr});\statement' \hastype 
\sendTyped{\objname}{\var}{\methodname}{\phi}.\T \has E
}
$

\vspace{2mm}
$
\TINFER{T-Call-2}{
\deduce{
\Phi \Rightarrow [\statement;\objname!\methodname(\many{\expr})]\phi
}{
\Phi, \statement;\objname!\methodname(\many{\expr}) \proves \statement' \hastype \T \has E
}
}{
\Phi, \statement \proves \objname!\methodname(\many{\expr});\statement' \hastype 
\sendTyped{\objname}{}{\methodname}{\phi}.\T \has E
}
$
$
\TINFER{T-Assign}{
\Phi, \statement;\type~\var = \expr \proves \statement' \hastype \T \has E'
}{
\Phi, \statement \proves \type~\var = \expr;\statement' \hastype \T \has E
}
$

\vspace{2mm}
$
\TINFER{T-Get}{
\deduce{
E = E' \cup \{(n,n') | \exists \methodname \in \mathsf{p2}(\expr).~ n \in \mathsf{term}(\methodname) \wedge n' \in \mathsf{node}(\statement;\expr.\getABS)\}
}{
\Phi, \statement;\type~\var = \expr.\getABS \proves \statement' \hastype \T \has E'
}
}{
\Phi, \statement \proves \type~\var = \expr.\getABS;\statement' \hastype 
\ReadAs{\expr}.\T \has E
}
$

\vspace{2mm}
$
\INFER{T-While}
{
    \deduce{
         \phi\wedge\mathsf{Post}(\Prcl),\xabs{skip} \proves \statement \hastype \T \has E' \qquad E = E' \cup E''
    }{
         \phi\wedge\mathsf{Post}(\Prcl),\xabs{skip} \proves \statement' \hastype \T' \has E'' \qquad \Phi \Rightarrow [\statement'']\phi \qquad \phi\wedge\mathsf{Post}(\Prcl)\Rightarrow [\statement]\phi
    }
}
{
    \Phi, \statement'' \proves \whileABS \expression\ \{\statement\}; \statement' \hastype (\T)^{\ast}_{\phi}\nnn\T' \has E
}
$



\section{Soundness}\label{sec:proof}
The proof for Theorem 1 is similar to the proof for Theorem 2 in~\cite{thesis}, we thus only give a sketch and point out where the proofs differ.

\subsection*{Propagation}
First, we state the correctness of propagation.
Let $\trace\upharpoonright_\objname$ be the projection of trace $\trace$ on $\objname$, i.e., $\trace\upharpoonright_\objname$ results from $\trace$ by replacing all events not issued by $\objname$.
\begin{lemma}\label{lem:prop}
Let $\prgm$ be a program and $\Prcl$ a type for $\prgm$.
If in all traces produced by $\prgm$, the order of invocation events is the same, then every trace that satisfies the translation of the propagated type iff it satisfies the translation of the original type:
\small\begin{align*}
\vdash\prgm:\Prcl \rightarrow \forall \trace.~\forall\objname.~\prgm\Downarrow \trace \rightarrow 
\big(\trace\upharpoonright_\objname \models \histype(\prop^{\ast}(\project{\Prcl}{}{\objname})) \leftrightarrow \trace\upharpoonright_\objname \models \histype(\project{\Prcl}{}{\objname}) \big)\\
\end{align*}\normalsize
\end{lemma}
\begin{proof}
We fix $\objname$ and denote $\project{\Prcl}{}{\objname}$ with $\LPrcl$.
We show this by induction on the number $n$ of applications of \prop for the fixpoint.
\begin{description}
\item[Induction Base, $n = 0$] Then $\prop^{\ast}(\LPrcl) = \LPrcl$ and the lemma holds trivially.
\item[Induction Step, $n = n' + 1$]
By induction hypothesis there is a type $\LPrcl' = \prop^{n'}(\LPrcl)$ such that the desired property holds. We make a case distinction on the applied case in the definition of \prop in its last application:
\begin{description}
\item[-] \textbf{Case 1} 
$\PutAs{\phi}.\receiveTyped{\methodname}{\psi} \rightsquigarrow \PutAs{\phi}.\receiveTyped{\methodname}{\psi \wedge \phi@\objname}$\\
In this case we have to show that the start of execution of method $\methodname$ the formula $\psi$ as holds. 
Let $\methodname'$ be the method whose termination action $\PutAs{\phi}$ is responsible for $\phi$.
By assumption, the order of invocation events is fixed and the program can be typed. Thus, there is no trace such that between the invocation action of $\methodname'$ and $\methodname$, there is another invocation event.
Thus, each trace $\trace$ that contains pairs of the form $(\invocrev(\objname,\future,\methodname),\config)$, for some $\future,\config$ s.t. $\config \models \psi$ contains this pair as part of a subtrace of the following form:
\[\Big[\big(\resolvev(\objname,\future',\methodname',\expr'),\config'\big),\big(\invocrev(\objname,\future,\methodname),\config\big)\Big]\]
for some $\future',\config$ s.t. $\config \models \phi$. 
Every state change on $\objname$ must be executed by some process on $\objname$, but as there is no such such process
(as there would be an invocation event for it) between the two events in the subtrace, the state-part of $\phi$, i.e. $\phi@\objname$, still holds at the invocation event: $\config \models \phi@\objname$.
This is exactly the condition captured by this propagation case.
\item[-] \textbf{Case 2} 
$\PutAs{\phi}\nnn(\T)^{\ast}_{\psi} \rightsquigarrow \PutAs{\phi \wedge \psi}\nnn(\T)^{\ast}_{\psi}$\\
By the definition of $\histype((\T)^{\ast})$ there is a set of indices $X$ in every trace, such that every such position $i\in X$, the invariant holds and for every pair of consecutive positions $i,j~\in X$,
the subtrace $\trace[i..j]$ satisfies $\histype(\T)$. Now, $\PutAs{\phi}$ is the last event before the repetition, thus before the very first position $i_0\in X$ there is a pair
\[\trace[i_0 - 1] = \big(\resolvev(\objname,\future,\methodname,\expr),\config\big)\]
In $\config$, no process is active at $\objname$. We only regard traces produced $\prgm$, thus $\trace$ is well-formed\footnote{The well-formedness of traces is defined in~\cite{Din14,thesis}} and $i_0$ must be a invocation reaction event
\[\trace[i_0] = \big(\invocrev(\objname,\future',\methodname'),\config\big)\]
Such that $\config \models \psi$. With the same argument as above, the condition at $\psi$ must hold at $i_0 - 1$, as there was no process who could have changed it.
Note, that $i_0 \neq 0$ as every local type starts with a receiving action, not a repetition.
\item[-] \textbf{Case 3} 
$(\T)^{\ast}_{\psi}\nnn\receiveTyped{\methodname'}{\phi} \rightsquigarrow (\T)^{\ast}_{\psi}\nnn\receiveTyped{\methodname'}{\phi\wedge\psi}$ \\ 
This case is analogous to case 2.
\item[-] \textbf{Case 4} 
$(\T)^{\ast}_{\phi}\nnn(\T)^{\ast}_{\psi} \rightsquigarrow (\T)^{\ast}_{\phi\wedge\psi}\nnn(\T)^{\ast}_{\psi}$\\
This case is analogous to case 2.
\item[-] \textbf{Case 5} 
$(\receiveTyped{\methodname'}{\phi}\nnn\T\nnn\PutAs{\phi'})^{\ast}_{\psi} \rightsquigarrow (\receiveTyped{\methodname'}{\phi\wedge\psi}\nnn\T\nnn\PutAs{\phi'\wedge\psi})^{\ast}_{\psi}$\\
By the definition of $\histype((\T)^{\ast})$, there is a set of indices $X$ in every trace, such that every such position $i\in X$, the invariant holds and for every pair of consecutive positions $i,j~\in X$,
If the repetition start with a receiving action and ends with a termination, the chosen positions are those of the termination actions and the first action before (Again, the syntactic form guarantees such a position). This reduces this case to show
that the same propagation as in case 1 holds and thus technical details are analogous to case 1.\qedhere
\end{description}
\end{description}

\end{proof}

\subsection*{Main Theorem}
Given a well-formed global type $\Prcl$ we can say that another (possibly not well-formed) global $\Prcl'$ is a prefix of $\Prcl$ if we can extend $\Prcl'$ to $\Prcl$ by concatenating another global type:
\[\Prcl' \sqsubseteq \Prcl \iff \Prcl = \Prcl'.\G\]
And similarly for local types $\LPrcl$ and traces $\trace$.

The main lemma is similar to subject reduction in non-model-theoretic semantics for types, as it connects types and operational semantics of the language.
It states that each step in the execution preserves the property that the trace so far is a prefix of a a trace which is a model for the type.
\begin{lemma}
Let $\prgm$ be a program and $\Prcl$ a well-formed type with $\vdash \prgm : \Prcl$. 
Every prefix of every trace of $\prgm$ satisfies the translation of a prefix of $\Prcl$:
\[\vdash \prgm : \Prcl \rightarrow \forall \trace.~\prgm\Downarrow\trace \rightarrow \Big(\forall \trace'.~\trace' \sqsubseteq \trace \rightarrow \big(\exists \Prcl'.~ \Prcl'\sqsubseteq \Prcl \wedge \trace' \models \histype(\Prcl')\big)\Big)\]
\end{lemma}
The property that the whole execution of the program satisfies the translation of the whole type and not some prefix follows from well-formedness of global types (Theorem 1 in~\cite{thesis}) and deadlock freedom.
The main differences to the proof in~\cite{thesis} are the following:

\begin{description}
\item[All assumed conditions hold at the point they are used] 
We distinguish between the following kinds of assumed conditions:
\begin{itemize}
\item The precondition at method start. The precondition is a conjunction $\phi_1\wedge\phi_2$ where $\phi_1$ is resulting from the projection and $\phi_2$ from the propagation.
That $\phi_2$ holds follows from Lemma~\ref{lem:prop}. That $\phi_1$ holds follows from the fact that the precondition is projected from a formula $\psi$ in the global call, which is fully proven 
by the caller, checked in rule $\rulename{T-Call}$ with the condition that $\psi$ is equal to its projection on the caller. 
\item The selection condition of the passive choice. It must connect that the additional condition executes the statement branch which is typed with the corresponding type branch.
This follows directly from the two additional promises of $\rulename{T-Offer}$ with respect to $\rulename{T-Select}$.
\end{itemize}
\item[Methods are executed in the right order] 
Assume there is a method $\methodname_1$ that is executed on some object $\objname$ before $\methodname_2$ in the type, but executed the other way around in the generated trace.
Then $\methodname_2$ does not depend on $\methodname_1$. But by assumption the program has been typed. Rule $\rulename{T-Main}$ checks, however, that the start of $\methodname_1$ causes $\methodname_2$ in its admissibility check
and this means that $\methodname_2$ cannot be executed before $\methodname_1$ (Lemma 36 in~\cite{thesis}).
\item[Deadlock Freedom] 
A deadlocked configuration is a configuration which is not terminated, yet cannot continue execution.
First we observe that every deadlock is caused by processes blocking at \abs{get} statements. It cannot be a single process, because a process has no access on its own future.
It can also not be stored in the heap between call and execution start, as this would mean that another method was active to store it and this would violate the condition that all methods are executed
in the right order shown above.

Assume there would be a deadlock. W.l.o.g. we assume that only two processes are involved, $p_1$ executing $\methodname_1$ and $p_2$ executing $\methodname_2$. 
If $p_1$ blocks while attempting to read a future belonging to $\methodname_2$ then
the Points-To analysis will include $\methodname_2$ in the set of possible method in rule $\rulename{T-Get}$. As $\methodname_1$ has been type checked, this mean that an edge from the corresponding termination to the corresponding read would be added to 
the causality graph. The same holds for $\methodname_2$. The termination of $\methodname_i$ is in the same object type as the reading type, 
thus there is a path from the read in $\methodname_i$ to the termination in $\methodname_i$.
The resulting graph is pictured below and contains a cycle. The absence of cycles is however checked in rule $\rulename{T-Main}$. For a full formalization of deadlocks through causality graphs, we refer to~\cite{MHPDead}.
\begin{center}
\begin{tikzpicture}[scale=1, every node/.style={scale=1}]
    \node[draw] at (0, 0)          (m1) {$\methodname_1$};
    \node[draw,circle] at (1, 0)   (m2) {\dots};
    \node[draw,circle] at (2, 0)   (m3) {$\uparrow$};
    \node[draw,circle] at (3, 0)   (m4) {\dots};
    \node[draw,circle] at (4, 0)   (m5) {$\downarrow$};

    \node[draw] at (0, 1)          (n1) {$\methodname_2$};
    \node[draw,circle] at (1, 1)   (n2) {\dots};
    \node[draw,circle] at (2, 1)   (n3) {$\uparrow$};
    \node[draw,circle] at (3, 1)   (n4) {\dots};
    \node[draw,circle] at (4, 1)   (n5) {$\downarrow$};

     \draw[->] (m2) -- (m3);
     \draw[->] (m3) -- (m4);
     \draw[->] (m4) -- (m5);

     \draw[->] (n2) -- (n3);
     \draw[->] (n3) -- (n4);
     \draw[->] (n4) -- (n5);

     \draw[->] (n5) -- (m3);
     \draw[->] (m5) -- (n3);

\end{tikzpicture}
\end{center}
\end{description}

\end{document}